\newif \iflipics
\lipicstrue

\iflipics

\documentclass[a4paper,UKenglish,cleveref,autoref,cequ]{lipics-v2021-modified}

\pdfoutput=1 
\hideLIPIcs  

\title{Distal combinatorial tools for graphs of bounded twin-width}

\author{Wojciech Przybyszewski}{Institute of Informatics, University of Warsaw, Poland \and \url{https://www.mimuw.edu.pl/~przybyszewski/}}{przybyszewski@mimuw.edu.pl}{https://orcid.org/
0000-0003-1158-9925}{}

\authorrunning{W. Przybyszewski} 

\Copyright{Wojciech Przybyszewski} 

\ccsdesc[100]{Mathematics of Computing $\to$ Discrete Mathematics $\to$ Graph Theory; Theory of computation $\to$ Logic $\to$ Finite Model Theory} 
\keywords{
    graph theory, twin-width, distality, neighborhood complexity,  Vapnik-Chervonenkis density, abstract cell decomposition, regularity lemma, cutting lemma

}

\funding{\ERCagreement}

\acknowledgements{
    I am deeply grateful to Prof. Szymon Toruńczyk, my master's thesis supervisor, for introducing me to this area of research as well as for his assistance and dedicated involvement in every step throughout the process of preparing this paper.
I am also grateful to the Anonymous Reviewers for their insightful comments, which helped me in rewriting this paper to its final form.

}

\nolinenumbers 

\EventEditors{John Q. Open and Joan R. Access}
\EventNoEds{2}
\EventLongTitle{42nd Conference on Very Important Topics (CVIT 2016)}
\EventShortTitle{CVIT 2016}
\EventAcronym{CVIT}
\EventYear{2016}
\EventDate{December 24--27, 2016}
\EventLocation{Little Whinging, United Kingdom}
\EventLogo{}
\SeriesVolume{42}
\ArticleNo{23}

\usepackage[T1]{fontenc}
\usepackage{amsmath}
\usepackage{amsthm}
\usepackage{mathtools}
\usepackage{stackengine}
\usepackage{bbm}
\usepackage{thm-restate}
\usepackage{hyperref}
\usepackage{breqn}

\newcommand{\ERCagreement}{{\begin{minipage}{.56\textwidth}This work is a part of project {\sc{BOBR}} that has received funding from the European Research Council (ERC) under the European Union's Horizon 2020 research and innovation programme (grant agreement No 948057). \end{minipage}\hfill\begin{minipage}{.33\textwidth}\includegraphics[width=\textwidth]{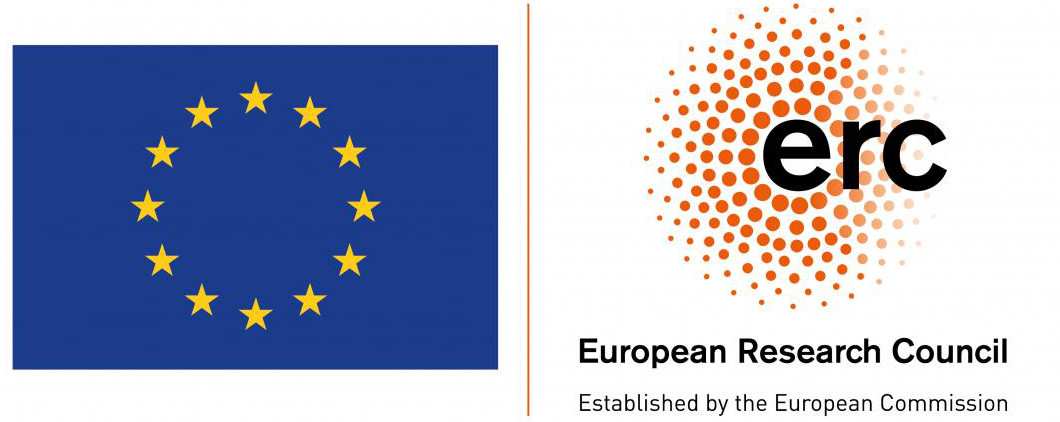}\end{minipage}\hfill}}

\newcommand{\A}{\mathbb A}

\newcommand{\R}{\mathbb R}

\newcommand{\N}{\mathbb N}
\newcommand{\F}{\mathcal F}
\newcommand{\C}{\mathcal C}
\newcommand{\Pp}{\mathcal P}

\newcommand{\T}{\mathcal T}
\newcommand{\I}{\mathcal I}
\newcommand{\M}{\mathbb M}
\newcommand{\Row}{\mathcal R}
\newcommand{\Col}{\mathcal C}

\newcommand{\nz}[1]{\widehat{N_{#1}}}
\newcommand{\set}[1]{\{#1\}}
\newcommand{\seq}[2]{#1_{1}, \ldots, #1_{#2}}
\newcommand{\seqi}[1]{1, \ldots, #1}

\renewcommand{\hat}{\widehat}
\renewcommand{\phi}{\varphi}
\renewcommand{\epsilon}{\varepsilon}
\renewcommand{\bar}{\overline}

\DeclareMathOperator{\tww}{tww}
\DeclareMathOperator{\vc}{vc}
\DeclareMathOperator{\prof}{Prof}

\DeclareMathOperator{\ass}{assert}

\newtheorem*{theorem*}{Theorem}

\begin{document}
\maketitle

\begin{abstract}
We study set systems formed by neighborhoods in graphs of bounded twin-width.
We start by proving that such graphs have linear neighborhood complexity, in analogy to previous results concerning graphs from classes with bounded expansion and of bounded clique-width.
Next, we shift our attention to the notions of distality and abstract cell decomposition, which come from model theory.
We give a direct combinatorial proof that the edge relation is distal in classes of ordered graphs of bounded twin-width.
This allows us to apply Distal cutting lemma and Distal regularity lemma, so we obtain powerful combinatorial tools for graphs of bounded twin-width.
\end{abstract}

\section{Introduction}
\label{chapter:introduction}
Twin-width is a graph width parameter recently introduced by Bonnet et al. in \cite{twinwidth1} with remarkable properties.
In particular, classes of bounded twin-width generalize some of the previously examined graph classes (e.g. planar graphs, graphs excluding a fixed minor, graphs of bounded clique-width), while admitting good structural properties (e.g. being closed under FO-transductions), algorithmical properties (e.g. FO model checking can be done in linear time on graphs of bounded twin-width given a contraction sequence) and combinatorial properties (e.g. $\chi$-boundedness \cite{twinwidth3}).

\medskip

\paragraph*{Graphs of bounded twin-width have linear neighborhood complexity}
Our focus is on the combinatorial and logical complexity of set systems defined by neighborhoods in graphs of bounded twin-width.
This continues a line of research studying other graph classes in this context.
Namely, graphs of bounded clique-width have linear neighborhood complexity, as it was proven in \cite{cliquewidth}.
That means that for every graph $G$ of clique-width at most $c$ and any non-empty subset $A$ of its vertices there are at most $n_c|A|$ different neighborhoods of vertices of $G$ in $A$ (where the constant $n_c$ depends only on $c$).
The same is also true for planar graphs and graphs excluding a fixed minor.
More generally, the same holds for classes with bounded expansion \cite{benc}, although they are not generalized by classes of bounded twin-width.
Our first result states that graphs of bounded twin-width also have linear neighborhood complexity.

\begin{theorem*}[Informal version of Theorem \ref{theorem:neighborhood-complexity-theorem}]
For every integer~$t$, there is a constant $n_t$ such that for every graph $G$ of twin-width at most $t$ and every non-empty subset $A \subseteq V(G)$, the vertices of $G$ have at most $n_t|A|$ different neighborhoods in~$A$.
\end{theorem*}

Our proof relies on properties of matrices of graphs of bounded twin-width shown in \cite{twinwidth1}, as well as the celebrated Stanley-Wilf conjecture/Marcus-Tardos theorem \cite{marcustardos}.
We remark that Theorem \ref{theorem:neighborhood-complexity-theorem} was proven independently by Bonnet et al. in \cite{twinwidthkernels}.
Although the techniques used there are similar to ours, we decided to include a self-contained proof of Theorem \ref{theorem:neighborhood-complexity-theorem} since the underlying ideas are needed in the later part of the paper.

Theorem \ref{theorem:neighborhood-complexity-theorem} in particular implies that the VC-density of graphs of bounded twin-width is equal to $1$ (Corollary \ref{cor:vc-density}).
This notion is a refinement of VC-dimension, which is the most well-known measure of complexity of set systems and was introduced in \cite{vcdim}.
Both measures are applied in a wide range of fields, including statistical learning theory and computational geometry.
Whereas a set system has bounded VC-dimension if and only if it has bounded VC-density, in general, VC-density is at most as large as the VC-dimension, and it turns out that the VC-density, not the VC-dimension, is the decisive measure for the combinatorial complexity of a family of sets \cite{VCdensity}.
For example, the VC-density of $\mathcal S$ governs the size of packings in $\mathcal S$ with respect to the Hamming metric \cite{VCpackings} and is intimately related to the notions of entropic dimension \cite{VCentropy} and discrepancy \cite{VCdiscrepancy}.
VC-density has also a number of applications in algorithmics, for example \textsc{Connected $k$-Vertex Cover} admits a kernel with $O(k^{1.5})$ vertices on classes of graphs with VC-density at most $1$ \cite{twinwidthkernels}.

\medskip
\paragraph*{The edge relation is distal in classes of graphs of bounded twin-width}
In the following part of the paper we refine our analysis of set systems defined by neighborhoods in graphs of bounded twin-width, and turn our attention to their logical complexity.
We do it by investigating the notion of distality defined by Simon in \cite{distality-definition}.
This notion comes from model theory and aims at describing theories which are NIP and \emph{purely instable}.
The original definition is phrased in the language of model theory and was stated for theories and single infinite models.
However, there is also an equivalent definition given in \cite{distal}.
It uses abstract cell decompositions, which are more combinatorial objects.
We adapt this definition to the context of classes of finite structures, e.g. graphs., by saying that a formula is distal in a class of structures if it admits an abstract cell decomposition.

It is known that every graph of bounded twin-width can be equipped with a linear order on its vertices such that the twin-width of the ordered graph (seen as a binary structure) is still bounded.
However, computing an optimal order is a hard problem.
In fact, even the problem of deciding whether a given (unordered) graph has twin-width at most $4$ is NP-complete \cite{twin-width-4-np-complete}.
On the other hand, when it comes to the problem of FO model checking on ordered graphs, it can be solved in fpt time precisely on these classes of ordered graphs that have bounded twin-width \cite{twinwidth4}. 
It is also known that classes of ordered graphs of bounded twin-width are distal by combining a few results from model theory.
However, this proof is highly unconstructive and relies on the compactness theorem for first-order logic.
In the main result of the paper (Theorem \ref{theorem:encoding_theorem}) we give a direct combinatorial proof that the edge relation is distal in classes of graphs of bounded twin-width, provided that we add to every graph an order that is a witness of a low twin-width.

\begin{theorem*}[Informal version of Theorem \ref{theorem:encoding_theorem}]
    Let $\widehat{\C}$ be a class of ordered graphs of twin-width bounded by some constant~$t$.
    Then, the formula $\phi(x; y) \equiv E(x, y)$ admits in $\widehat{C}$ a distal cell decomposition weakly definable by a finite set $\Psi(x; y_1, \ldots, y_k)$, where $k = \mathcal O(t)$.
    Moreover, the exact form of formulas in $\Psi$ follows from the proof.
\end{theorem*}

This allows us to e.g. obtain effective bounds on the constants in theorems about powerful distal combinatorial tools proven in \cite{regularity, distal} and to better understand structure of sets that appear in these theorems.

\medskip

\paragraph*{Distal combinatorial tools}

After adapting the notion of abstract cell decompositions to classes of finite structures we translate theorems about powerful combinatorial tools proven in \cite{regularity, distal} into this setting.
We now briefly describe these tools and motivations for investigating abstract cell decompositions.

\paragraph*{Cutting lemma}

The so-called cutting lemma is a very useful combinatorial partition tool with numerous applications in computational and incidence geometry and related areas (see e.g. \cite[Sections 4.5, 6.5]{matousek_lectures} or \cite{cuttings} for a survey).
In its simplest form it can be stated as follows (see e.g.
\cite[Lemma 4.5.3]{matousek_lectures}).
\begin{theorem*}[Cutting lemma]
    For every set $L$ of $n$ lines in the real plane and every $1 < r < n$ there exists a $\frac 1r$-cutting for $L$ of size $O(r^2)$.
    That is, there is a subdivision of the plane into generalized triangles (i.e. intersections of three half-planes) $\seq \Delta t$ so that the interior of each $\Delta_i$ is intersected by at most $\frac nr$ lines in $L$, and we have $t \le Cr^2$ for a certain constant $C$ independent of $n$ and $r$.
\end{theorem*}
This result provides a method to analyze intersection patterns in families of lines, and it has many generalizations to higher dimensional sets and/or to families of sets of more complicated shape than lines, for example families of algebraic or semialgebraic curves of bounded complexity \cite{cuttings_example}.
Combining the result of Chernikov, Galvin and Starchenko from \cite{distal} with Theorem \ref{theorem:encoding_theorem} we get a version of distal cutting lemma for graphs of bounded twin-width.
\begin{theorem*}[Informal version of Theorem \ref{theorem:cutting-lemma-tww}]
    Let $\C$ be a class of graphs of twin-width at most $t$.
    Then, there is an integer $d = \mathcal O(t)$ with the following property:
    for any graph $G \in \C$, any $A \subseteq V(G)$ of size $n$, and any real $1 \le r \le n$ we can partition the vertices of $V(G)$ into at most $Cr^{d}$ sets $\seq Xl$ such that the vertices in every $X_i$ have almost the same neighborhood in $A$.
    More precisely, there are at most $\frac nr$ vertices $a \in A$ for which there are $u, v \in X_i$ with $(u, a) \in E(G)$ and $(v, a) \not \in E(G)$.
\end{theorem*}
Let us remark that this theorem is obvious for $r = 1$ (we just take $X_1 = V(G)$) and the version for $r = n$ corresponds to having a polynomial neighborhood complexity and is implied by Theorem \ref{theorem:neighborhood-complexity-theorem}.

\paragraph*{Regularity Lemma}

Another property of distal classes of graphs is a stronger version of the regularity lemma.
The original Szemerédi Regularity Lemma is a fundamental result in graph combinatorics with many versions and applications in extremal combinatorics, number theory and computer science (see \cite{regularity-lemma} for a survey).
Roughly speaking, it says that for every $\epsilon > 0$ there is a bound $K = K(\epsilon)$ such that every graph $G$ can be partitioned into $K$ parts,
such that the edges between any two parts (apart from a few exceptions)
behave in a regular way (where $\epsilon$ controls the degree of regularity).
In its simplest form it can be presented as follows.

\begin{theorem*}[Regularity lemma]
    For every $\epsilon > 0$ there exists $K = K (\epsilon)$ such that: for every graph $G = (V, E)$ there exists a partition $V = V_1 \cup \ldots \cup V_k$ into non-empty sets and a set $\Sigma \subseteq [k] \times [k]$ with the following properties.
    \begin{enumerate}
        \item Bounded size of the partition: $k \le K$.
        \item Few exceptions: $|\bigcup_{(i, j) \in \Sigma} V_i \times V_j| \ge (1-\epsilon)|V|^2$.
        \item $\epsilon$-regularity: for all $(i, j) \in \Sigma$ and all $A \subseteq V_i, B \subseteq V_j$ with $|A| \ge \epsilon |V_i|, |B| \ge \epsilon |V_j|$, one has
        \[
            \big| d(A, B) - d(V_i, V_j) \big| \le \epsilon,
        \]
        where $d(X, Y) = \frac{|E(X, Y)|}{|X||Y|}$ and $E(X, Y)$ is the set of edges with one endpoint in $X$ and the other in $Y$.
    \end{enumerate}
\end{theorem*}

In general the bound on the size of the partition $K$ is known to grow as an exponential tower of height $\frac 1\epsilon$.
Recently several improved regularity lemmas were obtained in the
context of definable sets in certain structures or in restricted families of structures (see e.g. \cite{regularity-application-1, regularity-application-2}).

Chernikov and Starchenko proved much stronger version of the regularity lemma for distal structures \cite{regularity}.
For classes of graphs which are NIP (this is a property of classes of structures less restrictive than distality) and their edge relation is distal (in particular graphs of bounded twin-width satisfy this conditions), this can be translated as:

\begin{theorem*}[Informal version of Theorem \ref{theorem:simple_regularity_lemma}]
    Let $\C$ be an NIP class of graphs with a distal edge relation.
    Then, there is a constant $c$ depending only on $\C$ with the following property: for every $\epsilon > 0$ and for every graph $G \in \C$, there exists a partition $V(G) = V_1 \cup \ldots \cup V_k$ into non-empty sets, and a set $\Sigma \subseteq [k] \times [k]$ with the following properties.
    \begin{enumerate}
        \item Polynomially bounded size of the partition: $k \le c(\frac 1\epsilon)^{c}$.
        \item Few exceptions: $|\bigcup_{(i, j) \in \Sigma} V_i \times V_j| \ge (1-\epsilon)|V|^2$.
        \item $0-1$-regularity: for all $(i, j) \in \Sigma$ there are either all edges between $V_i$ and $V_j$ or no edge at all.
    \end{enumerate}
\end{theorem*}

\paragraph*{Structure of the paper}

The paper is structured as follows.
In Section \ref{chapter:preliminaries} we introduce the preliminary notions.
Then, the proof of Theorem \ref{theorem:neighborhood-complexity-theorem} is presented in Section \ref{chap:nc_proof}.
After the proof of Theorem \ref{theorem:neighborhood-complexity-theorem} we switch our attention to the notion of distality and abstract cell decompositions in Section \ref{chap:distality}.
In particular, we translate the theorems about distal combinatorial tools into the setting of classes of finite structures.
Then, in Section \ref{chap:encoding} we give a direct combinatorial proof that the edge relation is distal in classes of ordered graphs of bounded twin-width.
Finally, in Section \ref{chap:conclusions} we present conclusions that follow from our work.

\section{Preliminaries}
\label{chapter:preliminaries}

We denote by $[i]$ the set of integers $\set{\seqi i}$.
If $\mathcal X$ is a set of sets, we denote by $\bigcup \mathcal X$ the union of them.

\subsection{Graph definitions and notations}
In this paper we investigate only undirected simple graphs, i.e. graphs with no multiple-edges nor self-loops.
We denote by $V(G)$ the set of vertices of a given graph $G$ and by $E(G)$ the set of its edges.

For a graph $G = (V, E)$ and $A \subseteq V$ we denote by $G - A$ the graph $\set{V \setminus A, \set{(u, v) \in E : u, v \not \in A}}$, i.e. the graph obtained from $G$ by removing all the vertices in $A$ and edges incident to them.

For $A \subseteq V(G)$ and $v \in V(G)$ we denote the neighborhood of $v$ in $A$ by $N^G_A(v)$, i.e. $N^G_A(v) = \set{u \in A : (u, v) \in E(G)}$.
We omit a superscript $G$ whenever the graph is implicit from the context.
We write $N(v)$ for $N_{V(G)}(v)$.
Similarly, we denote the set of all neighborhoods in $A \subseteq V(G)$ by $N_G(A)$, i.e. $N_G(A) = \set{N_A(u) : u \in V(G)}$.
We also write $\nz G(A)$ for the set of all non-empty neighborhoods in $A$, i.e. $\nz G(A) = N_G(A) \setminus \set \emptyset$.

For a graph $G = (V, E)$ and two subsets of its vertices $A, B \subseteq V$ we say that $A$ and $B$ are \emph{homogeneous}, if either for every $v \in A$ and $u \in B$ we have $(v, u) \in E$ or for every $v \in A$ and $u \in B$ we have $(v, u) \not \in E$.
In particular, if $A \cap B \neq \emptyset$ then $A$ and $B$ can be homogeneous only if for every $v \in A$ and $u \in B$ we have $(v, u) \not \in E$.

\subsection{Matrix definitions and notations}
For a matrix $M$ consisting of $m$ rows and $n$ columns (an $m \times n$ matrix) we denote by $M[i][j]$ its entry in the $i$'th row and $j$'th column (assuming that $1 \le i \le m$ and $1 \le j \le n$).
A $p \times q$ submatrix $N$ of an $m \times n$ matrix $M$ (with $p \le m$ and $q \le n$) is any matrix formed by taking $p$ consecutive rows and $q$ consecutive columns of $M$.
For an $m \times n$ matrix $M$ we denote by $M[i:j][k:l]$ (with $1 \le i \le j \le m$ and $1 \le k \le l \le n$) the submatrix of $M$ formed by rows from $i$ to $j$ (inclusive) and columns from $k$ to $l$ (inclusive).

An $m \times n$ matrix $M$ is \emph{vertical} if every two rows of $M$ are equal, i.e. for any $1 \le i, j \le m$ we have $M[i:i][1:n]$ is equal to $M[j:j][1:n]$.
Similarly, $M$ is \emph{horizontal} if every two columns of $M$ are equal, i.e. for any $1 \le i, j \le n$ we have $M[1:m][i:i]$ is equal to $M[1:m][j:j]$.
Observe that if a matrix is both vertical and horizontal, then it is constant.
We say that a matrix is \emph{mixed} if it is neither vertical nor horizontal.
A \emph{corner} is a $2 \times 2$ mixed matrix.
In \cite{twinwidth1} it was proven that a matrix is mixed if and only if it contains a corner as a submatrix.

Given an $m \times n$ matrix $M$, a \emph{row-partition} (resp. \emph{column-partition}) is a partition of the rows (resp. columns) of $M$.
Similarly, a \emph{row-division} (resp. \emph{column-division}) is a row-partition (resp. column-partition), where every part consists of consecutive rows (resp. columns).
A \emph{$(k, l)$-division} (or simply division) of a matrix $M$ is a pair $(\Row, \Col)$ of a row-division and a column-division with respectively $k$ and $l$ parts.
A \emph{zone} of a division $(\Row, \Col) = (\set{\seq Rk}, \set{\seq Cl})$ is any submatrix $R_i \cap C_j$ for $1 \le i \le k$ and $1 \le j \le l$.

A $0, 1$-matrix is a matrix with all its entries equal to $0$ or $1$.
Given a $0, 1$-matrix $M$, a \emph{$t$-grid minor} in $M$ is a $(t, t)$-division of $M$ in which every zone contains a $1$.
Given a $0, 1$-matrix $M$, a \emph{$t$-mixed minor} in $M$ is a $(t, t)$-division of $M$ in which every zone is a mixed submatrix of $M$.
A matrix is \emph{$t$-grid free} (resp. \emph{$t$-mixed free}) if it does not contain a $t$-grid minor (resp. $t$-mixed minor).

If $G$ is an $n$-vertex graph and $\sigma$ is a total ordering of $V(G)$, say, $\seq vn$, then $M_\sigma(G)$ denotes the adjacency matrix of $G$ in the order $\sigma$.
Thus $M_\sigma(G)[i][j]$ is $1$ if $(v_i, v_j) \in E(G)$ and $0$ otherwise.
An \emph{ordered graph} is a pair $(G, \preceq)$ such that $G$ is a graph and $\preceq$ is a total order on $V(G)$.
We say that an ordered graph $(G, \preceq)$ is $t$-mixed free if $M_\preceq(G)$ is $t$-mixed free.
A first-order formula of ordered graphs is a standard first-order formula of graphs that can use additional $\preceq$ symbol, which is interpreted as an order of an ordered graph.

\subsection{Definition of twin-width of graphs}
\label{subsec:twin_width_graphs_def}
The twin-width of a given graph was first defined in \cite{twinwidth1}.
Since then, many tutorials on the topic were given, e.g. \cite{bonnet-tww, pilipczuk-fas} can serve as a good introduction.

The original definition of twin-width uses the notion of a \emph{trigraph}, i.e. a triple $G = (V, E, R)$ where $E$ and $R$ are two disjoint sets of edges on $V$: the (usual) edges and the red edges.
A trigraph $(V, E, R)$ such that $(V, R)$ has maximum degree at most $d$ is a $d$-trigraph.
Any graph $(V, E)$ may be interpreted as the trigraph $(V, E, \emptyset)$.

Given a trigraph $G = (V, E, R)$ and two vertices $u, v$ in $V$, we define the trigraph $G/u, v = (V', E', R')$ obtained by identifying $u, v$ into a new vertex $w$ as the trigraph on a new vertex-set $V' = (V \setminus \set{u, v}) \cup \set w$ such that $G - \set{u, v} = (G/u, v) - \set w$ and the following edges incident to $w$:
\begin{itemize}
    \item $(w, x) \in E'$ if and only if $(u, x) \in E$ and $(v, x) \in E$,
    \item $(w, x) \not \in E' \cup R'$ if and only if $(u, x) \not \in E$ and $(v, x) \not \in E$,
    \item $(w, x) \in R'$ otherwise.
\end{itemize}
We say that $G/u, v$ is a \emph{contraction} of $G$.
If both $G$ and $G/u, v$ are $d$-trigraphs, $G/u, v$ is a $d$-contraction.

A (tri)graph $G$ is $d$-collapsible if there exists a sequence of $d$-contractions which contracts $G$ to a single vertex.
The minimum $d$ for which $G$ is $d$-collapsible is the \emph{twin-width} of $G$, denoted $\tww(G)$.

\subsection{Theorems related to twin-width of graphs}
The following grid minor theorem for twin-width proven in \cite{twinwidth1} describes a connection between twin-width of a graph and mixed minors of its matrix.
\begin{theorem}
    \label{theorem:grid-minor-theorem}
    If $G$ is a graph of twin-width less than $t$, then there is a total ordering on its vertices $\sigma$ such that $M_\sigma(G)$ is $2t + 2$-mixed free.
    On the other hand, if $G$ is a graph and $\sigma$ is a total ordering on its vertices such that $M_\sigma(G)$ is $k$-mixed free, then $\tww(G) = 2^{2^{O(k)}}$.
\end{theorem}
A crucial tool for this proof is the Marcus-Tardos theorem about grid minors.
\begin{theorem}[\cite{marcustardos}]
    \label{theorem:marcus-tardos}
    For every integer $t$, there is some $c_t$ such that every $m \times n$ $0, 1$-matrix $M$ with at least $c_t \max(m, n)$ entries $1$ has a $t$-grid minor.
\end{theorem}
Marcus and Tardos established Theorem \ref{theorem:marcus-tardos} with $c_t = 2t^4 \binom{t^2}t$.
Then Cibulka and Kynčl \cite{constants} decreased $c_t$ to $8/3(t+1)^22^{4t}$.
In the rest of this paper we keep the notation $c_t$ for these constants.

\subsection{Vapnik-Chervonenkis density of graphs}
In this section we introduce notions of Vapnik-Chervonenkis dimension (VC-dimension) and Vapnik-Chervonenkis density (VC-density) in the context of graphs.
It is worth noting that it was first introduced in \cite{vcdim} and originated in statistics, but it is also an important notion in combinatorics and statistical learning theory.

We say that $A \subseteq V(G)$ is \emph{shattered} if $N_G(A) = 2^A$, i.e. for every subset of $A$ there is a vertex $v \in V(G)$ such that $N_A(v)$ is precisely that subset of $A$.
For a graph $G$ we define its \emph{VC-dimension} as the size of the largest subset of $V(G)$, which is shattered, and denote it by $\dim(G)$.
We denote the \emph{shatter function} of a graph $G$ by $\pi_G: \N \to \N$ and define it as follows:
\[
    \pi_G(k) = \max_{A \subseteq V(G), |A| \le k} |N_G(A)|.
\]
For a non-empty class of graphs $\C$ we define its VC-dimension as the supremum over VC-dimensions of graphs in $\C$ and denote it by $\dim(\C)$.
Similarly, we denote its shatter function by $\pi_\C: \N \to \N$ and define it by:
\[
    \pi_\C(k) = \max \set{\pi_G(k): G \in \C}.
\]
Observe that this is well-defined, as $\pi_G(k) \le 2^k$ for any graph $G$ and $k \in \N$.

We can bound the shatter function of a graph class $\C$ in terms of $\dim(\C)$.
This is stated in the shatter function lemma proven independently by Sauer \cite{sauer} and Shelah \cite{shelah}.
\begin{lemma}[Shatter function lemma]
    \label{lem:shatter-function-lemma}
    If $\C$ is a graph class such that $\dim(\C) \le d$, then $\pi_\C(k) = O(k^d)$.
\end{lemma}
 For a class of graphs $\C$ we also define its \emph{VC-density} (denoted by $\vc(\C)$) as follows:
\iflipics
\[
    \vc(\C) =
    \begin{cases}
        \inf \set{r \in \R : r > 0, \pi_\C(k) = O(k^r)} &\text{if $\dim(\C)$ is finite,}\\
        +\infty &\text{otherwise.}
    \end{cases}
\]
\else
\[
    \vc(\C) =
    \begin{cases}
        \inf \set{r \in \R : r > 0, \pi_\C(k) = O(k^r)} &\parbox{1.75cm}{if $\dim(\C)$ is finite,}\\
        +\infty &\text{otherwise.}
    \end{cases}
\]
\fi

\section{Neighborhood complexity of classes of graphs of bounded twin-width}
\label{chap:nc_proof}

In this section we prove that classes of graphs of bounded twin-width admit linear neighborhood complexity.
This is formalized in Theorem \ref{theorem:neighborhood-complexity-theorem}.

\begin{theorem}
    \label{theorem:neighborhood-complexity-theorem}
    For every integer $t$, there is some $n_t$ such that for every graph $G$ of twin-width at most $t$ and every non-empty $A \subseteq V(G)$ we have $|N_G(A)| \le n_t|A|$.
\end{theorem}

The following definition and two lemmas will be useful for the proof.

\begin{definition}
\label{definition:corner-matrix}
Let $A$ be an $m \times n$ $0, 1$-matrix, where $m, n \ge 2$.
We define \textbf{the corner matrix $B$ of matrix $A$} as the $(m - 1) \times (n - 1)$ matrix given by:
\[
B[i][j] = \begin{cases}
    1 &\text{if $M[i:(i+1)][j:(j+1)]$ is a corner,} \\
    0 &\text {otherwise.}
\end{cases}
\]
\end{definition}

An example of a corner matrix is presented in Figure \ref{fig:corner_matrix}.

\begin{figure}
    \centering
    \includegraphics[scale=0.5,page=1]{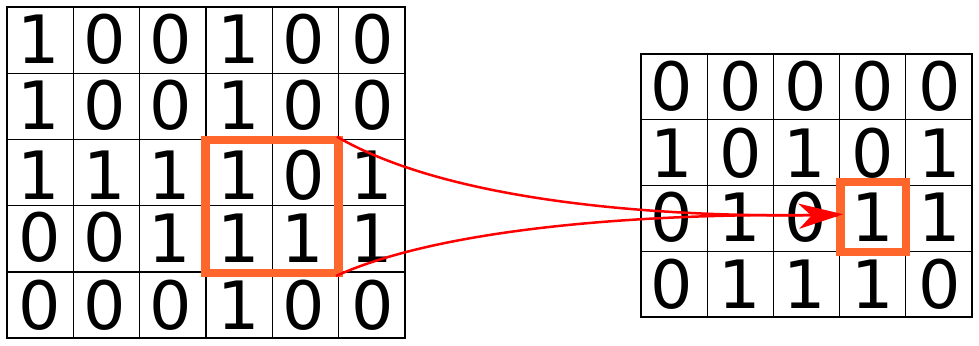}
    \caption{Example of a corner matrix}
    \label{fig:corner_matrix}
\end{figure}

\begin{lemma}
\label{lemma:corner-matrix-grid-free}
Let $A$ be a $0, 1$-matrix, which is $t$-mixed free.
Then its corner matrix $B$ is $2t$-grid free.
\end{lemma}
\begin{proof}
Let us assume by contradiction, that $B$ admits a $(2t, 2t)$-division $(\Row, \Col) = (\set{\seq R{2t}},\linebreak[1] \set{\seq C{2t}})$, which is a $2t$-grid minor.
Every part $R_k$, which consists of $l$ rows of $B$, corresponds to $l+1$ consecutive rows of $A$.
Precisely, if $R_k$ consists of rows $i$ to $i + l - 1$ of $B$, then entries in these rows depend only on rows $i$ to $i + l$ of $A$.
Moreover, whenever $|k - m| > 1$, $R_k$ and $R_m$ correspond to disjoint sets of rows of $A$ -- in particular $R_1, R_3, \ldots, R_{2t - 1}$ correspond to $t$ disjoint sets of rows of $A$.
The same happens for $C_1, C_3, \ldots, C_{2t-1}$.
These sets of rows and columns induce a $t$-mixed minor of $A$.
This is a contradiction to $A$ being $t$-mixed free.
\end{proof}

\begin{lemma}
\label{lemma:bounding-number-of-columns}
Let P be an $m \times n$ $0, 1$-matrix such that all corners in $P$ appear only in $p$ different pairs of consecutive rows, i.e. there are $p$ distinct integers $1 \le \seq ip < m$ such that every corner in $P$ is of a form $P[i_j:(i_j+1)][k:(k+1)]$ for some $1 \le j \le p$ and $1 \le k < n$.
Then $P$ has at most $2^{p+1}$ different columns.
\end{lemma}
\begin{proof}
We proceed by induction on $p$.
In the base case, i.e. when $p = 0$, $P$ does not contain any corner, so it is either horizontal or vertical.
If $P$ is horizontal, then all its columns are equal by definition.
On the other hand, when $P$ is vertical, then every column of $P$ contains only $0$s or only $1$s.
Therefore, $P$ contains at most $2$ different columns.

Now let us assume $p \ge 1$ and the statement is true for $p - 1$.
Without loss of generality $1 \le i_1 < \ldots < i_p < m$.
We can split $P$ horizontally into two submatrices $P_1 = P[1:i_p][1:n]$ and $P_2 = P[(i_p + 1):m][1:n]$.
All the corners of $P_1$ are contained in $p - 1$ different pairs of consecutive rows, so it contains at most $2^p$ different columns (see Figure \ref{fig:different-columns-lemma} for an illustration).
Moreover, $P_2$ does not contain any corner at all, so it contains at most $2$ different columns.
Every column of $P$ is made up using one column of $P_1$ and one column of $P_2$, so there are at most $2^p \cdot 2 = 2^{p+1}$ different columns in $P$.
\end{proof}

\begin{figure}
    \centering
    \iflipics
    \includegraphics[scale=0.8,page=2]{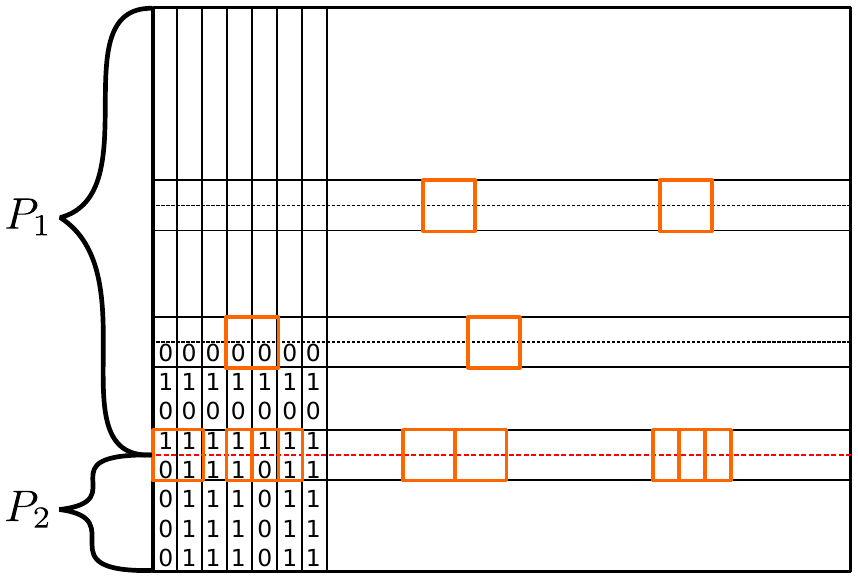}
    \else
    \includegraphics[scale=0.55,page=2]{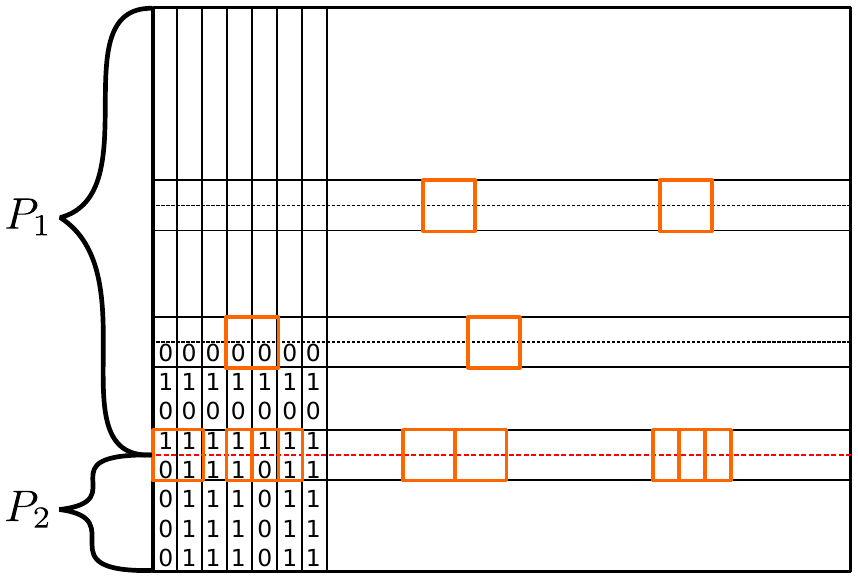}
    \fi
    \caption{A matrix $P$ with all its corners (in orange) contained in three pairs of consecutive rows is split into $P_1$ and $P_2$ in the inductive step.}
    \label{fig:different-columns-lemma}
\end{figure}

We also state an easy corollary which follows directly from the proof of Lemma \ref{lemma:bounding-number-of-columns} and will be useful in later sections.

\begin{corollary}
    \label{cor:representative_set}
    Let P be an $m \times n$ $0, 1$-matrix such that all corners in $P$ appear only in $p$ different pairs of consecutive rows.
    Then, there is a subset $I \subseteq [m]$ of at most $p + 1$ rows such that whenever two columns $j$ and $j'$ satisfy $P[i][j] = P[i][j']$ for every $i \in I$, then they are equal.
\end{corollary}

Now we prove that graphs of bounded twin-width have linear neighborhood complexity.

\begin{proof}[Proof of Theorem \ref{theorem:neighborhood-complexity-theorem}]
Let us take any graph $G$ such that $\tww(G) \le t$ and a non-empty subset $A \subseteq V(G)$ of vertices of $G$.
By Theorem \ref{theorem:grid-minor-theorem} there is a total ordering $\sigma$ on $V(G)$ such that $M_\sigma(G)$ is $2t+2$-mixed free.
For the rest of this proof we write $M$ for $M_\sigma(G)$.

Let $B \subseteq V(G)$ be a minimal subset of vertices of $G$ such that they represent all possible neighborhoods in $A$ apart from $\emptyset$.
Formally, $B$ is a minimal subset of $V(G)$ such that $\nz G(A) = \set{N_A(b): b \in B}$.
Obviously $|N_G(A)| \le |B| + 1$ so we need to bound the size of $|B|$.
Let us also remark, that for any distinct $b, b' \in B$ we have $N_A(b) \neq N_A(b')$, as in the opposite case we could remove $b'$ from $B$ obtaining a smaller set with the property we want.
Similarly, for any $b \in B$ we have $N_A(b) \neq \emptyset$.

Let us denote by $N$ the matrix that is obtained from $M$ by removing rows that correspond to vertices in $V(G) \setminus A$ and columns that correspond to vertices in $V(G) \setminus B$.
Obviously, $N$ is an $|A| \times |B|$ $2t+2$-mixed free matrix, as we obtain it by deleting some rows and columns of $M$, which is $2t+2$-mixed free itself.
We can see that any two consecutive columns of $N$ form a mixed matrix.
Indeed, they cannot be equal, so the matrix formed by them is not vertical.
It is not horizontal either, as there is no column in $N$ containing only zeros.
Therefore, for any two consecutive columns of $N$ there is a corner contained in them.
Let us denote by $C$ the corner matrix of $N$ (we assume $|A|, |B| \ge 2$ as the opposite case is trivial).
From what we observed so far, we know that in every column of $C$ there is at least one entry with $1$.

By Lemma \ref{lemma:bounding-number-of-columns} applied for $p = c_{4t+4}$ we obtain that every set of $2^{c_{4t+4}}$ consecutive columns of $C$ has at least $c_{4t+4}$ non-zero rows.
This is because every such set corresponds to a set of $2^{c_{4t+4}} + 1$ consecutive columns in $N$ (and columns of $N$ are pairwise different).

Suppose $C$ has at least $2^{c_{4t+4}} (|A| - 1)$ columns.
Pick any partition of its columns into ${|A| - 1}$ disjoint sets $(\seq B{|A| - 1})$, of at least $2^{c_{4t+4}}$ consecutive columns each.
Consider the $(|A| - 1) \times (|A| - 1)$ matrix $C'$ given by:

\iflipics
\[
C'[i][j] = \begin{cases}
    1 & \text{if there is a column in $B_j$ that has entry~$1$ in $i$'th row,}\\
    0 &\text{otherwise.}
\end{cases}
\]
\else
\[
C'[i][j] = \begin{cases}
    1 & \parbox{5cm}{if there is a column in $B_j$ that has entry~$1$ in $i$'th row,}\\
    0 &\text{otherwise.}
\end{cases}
\]
\fi
The matrix $C'$ has at least $c_{4t+4}$ entries with $1$ in every column.
Therefore, by Theorem \ref{theorem:marcus-tardos}, it admits a $4t+4$-grid minor, which trivially induces a $4t+4$-grid minor of $C$.
This is a contradiction with Lemma \ref{lemma:corner-matrix-grid-free}, which states that $C$ is $4t+4$-grid free.
Therefore $C$ has at most $2^{c_{4t+4}} (|A| - 1) - 1$ columns, so $G$ has at most $2^{c_{4t+4}} (|A| - 1) + 1$ neighborhoods in $|A|$.
Bearing in mind the separate case for $|A| = 1$, we can set $n_t = 2^{c_{4t+4}} = 2^{2^{O(t)}}$.
\end{proof}

\begin{corollary}
\label{cor:vc-density}
Let $\C_t$ be a class of all graphs of twin-width at most $t$.
Then $\vc(\C_t) = 1$.
\end{corollary}
\begin{proof}
Theorem \ref{theorem:neighborhood-complexity-theorem} trivially implies $\vc(\C_t) \le 1$, so it suffices to show $\vc(\C_t) \ge 1$.
By $G_k$ we denote the graph on vertices $[2k]$ with edges $\set{(2i-1, 2i): i \in [k]}$, i.e. $G_k$ is a sum of $k$ non-incident edges.
Let us also take $A_k = \set{2i: i \in [k]}$.
Obviously $G_k$ has twin-width $0$, so $G_k \in \C_t$ for any $t \ge 0$.
On the other hand, $|N_{G_k}(A_k)| = |A_k| + 1$, which implies $\vc(\C_t) \ge 1$.
\end{proof}

\section{Distality and distal combinatorial tools}
\label{chap:distality}

In this section we concentrate on the notion of \emph{distality}.
The definition of a distal theory was introduced in \cite{distality-definition} using notions from model theory. Later in \cite{distal}, an equivalent definition of a distal model was presented, this time using \emph{abstract cell decompositions}, which are more combinatorial objects.
To our knowledge, only the definitions of a distal theory and a distal model are explicitly present in the literature.
In this section, we also define notions of a distal class of structures and a distal formula.
We remark, that they are just straightforward generalizations of the previous definitions and they already appeared in the literature implicitly (e.g. \cite[Theorem 3.2]{distal}).
Using our definitions we reformulate theorems about powerful combinatorial tools presented in \cite{distal} and \cite{regularity}.
Originally, they were stated for single, infinite distal structures, but we derive analogous statements for the case of classes of (possibly finite) structures.
However, we need to start with additional preliminaries on model theory and introduce the notion of \emph{NIP} (standing for \emph{not the independence property}), which is closely related to distality.

\subsection{Model theory preliminaries}
In the remaining part of this paper we investigate only first-order structures over relational signatures.
In particular we treat a graph $G$ as a structure over the signature consisting only of the binary edge relation symbol $E$.
Similarly, we treat an ordered graph $(G, \preceq)$ as a structure over the signature consisting of two binary symbols -- the edge relation symbol $E$ and the order relation symbol~$\preceq$.

Our notations are standard.
We usually denote first-order structures by blackboard bold letters $\M, \N$, etc., and use letters $M, N$, etc.~to denote their domains.
We denote elements of domains and variables in formulas by small letters $a, x$, etc., and use small letters with a bar $\bar a, \bar x$, etc. to denote tuples of elements or variables.
Then we use $|\bar x|$ to denote the length of the tuple $\bar x = (\seq xn)$.

We sometimes treat a tuple $\bar x$ as a finite set of variables, and then we write $\phi(\bar x)$ to denote a first-order formula $\phi$ with free variables contained in $\bar x$.
We may also write $\phi(\bar x_1,\ldots,\bar x_k)$ to denote a formula whose free variables are contained in $\bar x_1 \cup \ldots \cup \bar x_k$.
We write $x$ instead of $\set x$ in case of a singleton set of variables, e.g. $\phi(x,y)$ always refers to a formula with two free variables $x$ and $y$.
We sometimes write $\phi(\bar x;\bar y)$ (with a semicolon between tuples) to distinguish a partition of the set of free variables of $\phi$ into two parts, $\bar x$ and $\bar y$; this partition plays an implicit role in some definitions.
If we have a set $\Psi$ of formulas with the same free variables $\bar x$ (i.e. $\Psi = \set{\phi_1(\bar x), \phi_2(\bar x), \ldots}$), then we may denote this set by $\Psi(\bar x)$ to emphasize the set of free variables.

If $\M$ is a structure and $\phi(\bar x; \bar y)$ is a formula in the language of $\M$, then for $\bar a \in M^{|\bar y|}$, by $\phi(M; \bar a)$ we denote the subset of $M^{|\bar x|}$ defined by $\phi(\bar x; \bar a)$, namely $\phi(M; \bar a) = \set{\bar b \in M^{|\bar x|} : \M \models \phi(\bar b; \bar a)}$.

For a given signature $\Sigma$ we say that a $\Sigma$-structure $\A$ is a \emph{binary structure}, if $\Sigma$ is a signature consisting only of binary relational symbols.

A \emph{theory} $T$ (over $\Sigma$) is a set of $\Sigma$-sentences.
A \emph{model of} a theory $T$ is a $\Sigma$-structure $\M$ such that $\M \models \phi$ for all $\phi \in T$.
When a theory has a model, it is said to be \emph{consistent}.
The \emph{theory of} a class of $\Sigma$-structures $\C$ is the set of all $\Sigma$-sentences $\phi$ such that $\M \models \phi$ for all $\M \in \C$.
The \emph{elementary closure} $\bar \C$ of $\C$ is the set of all models $\M$ of the theory of $\C$. 
Thus $\C \subset \bar \C$, and $\C$ and $\bar \C$ have equal theories. 

An important tool for checking if a theory is consistent is the compactness theorem.

\begin{theorem}[Compactness]
    \label{thm:compactness}
    A theory $T$ is consistent if and only if every finite subset $T'$ of $T$ is consistent.
\end{theorem}

\subsection{NIP and monadic NIP}
The notion of \emph{NIP} formulas was introduced by Shelah in his work on the classification program \cite{shelah-nip}.
We present here a definition from \cite{nip-definition}, which is equivalent to the original one.
\begin{definition}
    A formula $\phi(\bar x; \bar y)$ with two tuples of free variables $\bar x$ and $\bar y$ is NIP (or dependent) in a class of structures $\C$ if there is a constant $k$, such that for every structure $\M \in \C$ there are no tuples $\bar {a_i} \in M^{|\bar x|} (i \in [k])$, $\bar{b_I} \in M^{|\bar y|} (I \subseteq [k])$ satisfying
    \begin{align}\label{eq:nip_condition}
        \M \models \phi(\bar{a_i}; \bar{b_I}) \iff i \in I. \tag{$\ast$}
    \end{align}
\end{definition}

Moreover, we say that a class of structures $\C$ is NIP if every formula $\phi(\bar x; \bar y)$ is NIP in $\C$ and we say that a single structure $\M$ is NIP if the class $\set{\M}$ is.

Observe, that for a formula $\phi(\bar x; \bar y)$ and a constant $k$ we can write a first-order sentence $\ass(\phi(\bar x; \bar y), k)$ which states, that in a given structure $\M$ there are no tuples $\bar {a_i} \in M^{|\bar x|} (i \in [k])$, $\bar{b_I} \in M^{|\bar y|} (I \subseteq [k])$ satisfying condition (\ref{eq:nip_condition}), namely
\iflipics
\[
    \ass(\phi(\bar x; \bar y), k) \equiv \neg \left[ \exists_{\bar x_1} \ldots \exists_{\bar x_k}\exists_{\bar y_\emptyset} \ldots \exists_{\bar y_{[k]}} \\ \left( \bigwedge_{i \in I} \phi(\bar x_i; \bar y_I) \land \bigwedge_{i \not \in I} \neg \phi(\bar x_i; \bar y_I) \right) \right].
\]
\else
\begin{dmath*}
    \ass(\phi(\bar x; \bar y), k) \equiv \neg \left[ \exists_{\bar x_1} \ldots \exists_{\bar x_k}\exists_{\bar y_\emptyset} \ldots \exists_{\bar y_{[k]}} \\ \left( \bigwedge_{i \in I} \phi(\bar x_i; \bar y_I) \land \bigwedge_{i \not \in I} \neg \phi(\bar x_i; \bar y_I) \right) \right].
\end{dmath*}
\fi
Therefore, if $\phi(\bar x; \bar y)$ is NIP in a class of structures $\C$, then for some $k$ the formula $\ass(\phi(\bar x; \bar y), k)$ is true in all structures in $\C$. 
It follows that $\phi(\bar x; \bar y)$ is NIP in $\C$ if and only if it is NIP in its elementary closure $\bar \C$.
Moreover, $\C$ is an NIP class of structures if and only if $\bar \C$ is.

We remark that the definition of NIP formulas can be equivalently expressed in terms of VC-dimension.
Indeed, it is easy to see that a formula $\phi(\bar x; \bar y)$ is NIP in a class of structures $\C$ if and only if there is an integer $k$ such that for every $\M \in \C$ the family of sets $\set{\phi(M; \bar a) : \bar a \in M^{|y|}}$ has VC-dimension at most $k$.
This gives an equivalent definition of NIP using the notion of VC-dimension.

Let us also note at this point that a class of graphs of bounded twin-width together with compatible orders is NIP.
We state it formally in the following lemma.
\begin{lemma}\label{lemma:nip-twin-width}
    Let $\widehat{\C}$ be a class of ordered graphs with the following property: there is a constant $t \in \N$ depending only on $\widehat{\C}$ such that for every $(G, \preceq) \in \widehat{\C}$ the matrix $M_\preceq(G)$ is $t$-mixed free.
    Then, $\widehat{\C}$ is NIP.
\end{lemma}

We defer the proof of Lemma \ref{lemma:nip-twin-width} to 
\iflipics
Appendix~\ref{appendix:twin-width-mon-nip}.
\else
the full version of this paper.
\fi
The reason is that although the proof is a simple combination of \cite[Theorem 14]{twinwidth1}, \cite[Theorem 11]{twinwidth4}, and \cite[Theorem 3]{twinwidth4}, it requires introducing a number of new definitions which are irrelevant for the rest of this paper.

\subsection{Abstract cell decompositions and definitions of distality}
\label{subsec:abstract_cell_decompositions}
In this subsection we define the notion of an \emph{abstract cell decomposition} for a given formula in a given structure by following \cite{distal}.
This requires a number of new definitions -- in order to make it easier to understand them, we interleave new definitions with examples.
Finally, we use abstract cell decompositions for defining the notion of distality.

Let $\M$ be a fixed $\Sigma$-structure.
For sets $A, X \subseteq M^d$ we say that $A$ \emph{crosses} $X$ if both $X \cap A$ and $X \cap (M^d \setminus A)$ are nonempty.

For a formula $\phi(\bar x; \bar y)$ and a set $S \subseteq M^{|\bar y|}$ we say that a subset $A \subseteq M^{|\bar x|}$ is $\phi(\bar x; S)$-\emph{definable} if $A = \phi(M; \bar s)$ for some $\bar s \in S$.
For a set $A \in M^{|\bar x|}$ we say that $\phi(\bar x; S)$ crosses $A$ if $\phi(M; \bar s)$ crosses $A$ for some $\bar s \in S$.

\begin{example}\label{ex:crosses}
    Consider the structure $\M$ with domain $M = \set{v_1, v_2, \ldots}$ over the empty signature $\Sigma = \emptyset$ and a formula $\phi(x; y) \equiv x = y$.
    Then, for $S_k = \set{v_1, \ldots, v_k}$ only $\set{v_1}, \ldots, \set{v_k}$ are $\phi(x; S)$-definable.
    Moreover, $\phi(x; S_k)$ crosses $\set{v_1, v_2}$ (because $\set{v_1}$ does), but it does not cross $M \setminus S_k$.
    \qed
\end{example}

Given a finite set $S \subseteq M^{|\bar y|}$, a finite family $\F$ of subsets of $M^{|\bar x|}$ is called an \emph{abstract cell decomposition} for $\phi(\bar x, S)$ if $M^{|\bar x|} = \bigcup \F$ and every $\Delta \in \F$ is not crossed by $\phi(\bar x; S)$.
Observe, that we do not require the sets in $\F$ to be disjoint -- we just want them to cover the whole of $M^{|\bar x|}$.
An abstract cell decomposition for $\phi(\bar x; \bar y)$ is an assignment $\T$ that to each finite set $S \subseteq M^{|\bar y|}$ assigns an abstract cell decomposition $\T(S)$ for $\phi(\bar x; S)$.

Observe that every $\phi(\bar x; \bar y)$ admits an obvious abstract cell decomposition, with $\T(S)$ consisting of the atoms in the Boolean algebra generated by the $\phi(\bar x; S)$-definable sets (see Figure \ref{fig:atomic-acl}).
In general, defining these cells would require longer and longer formulas when $S$ grows, and we want to avoid this possibility.
Therefore, we say that an abstract cell decomposition $\T$ for $\phi(\bar x; \bar y)$ is \emph{weakly definable} if there is a finite set of formulas $\Psi(\bar x; \bar {y_1}, \ldots, \bar {y_k})$ with $|\bar y| = |\bar {y_1}| = \ldots = |\bar {y_k}|$ such that for any finite $S \subseteq M^{|\bar y|}$, every $\Delta \in \T(S)$ is $\Psi(\bar x; S^k)$-definable (i.e. $\Delta = \psi(M; \bar {s_1}, \ldots, \bar {s_k})$ for some $\bar {s_1}, \ldots, \bar {s_k} \in S$
and $\psi \in \Psi$).
In this case we say that $\Psi(\bar x; \bar {y_1}, \ldots, \bar {y_k})$ weakly defines $\T$.

\begin{figure}
    \centering
    \iflipics
    \includegraphics[scale=1.0]{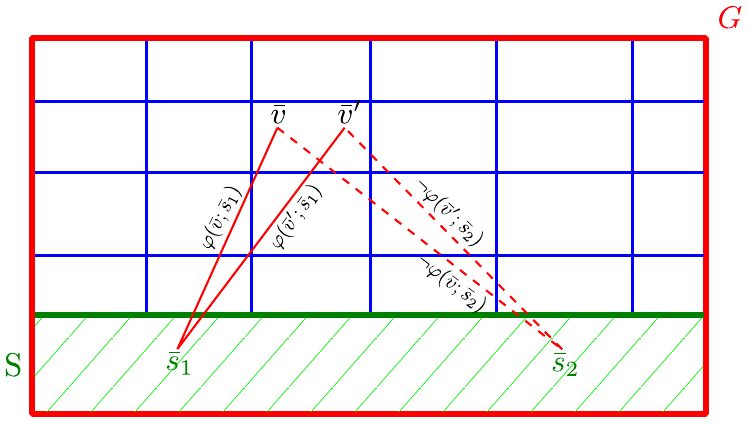}
    \else
    \includegraphics[scale=0.7]{images/atomic_acl.pdf}
    \fi
    \caption{For a given finite set $S \subseteq M^{|y|}$ we can divide $M^{|x|}$ into at most $2^{|S|}$ parts such that tuples in every part satisfy $\phi$ with the same elements of $S$. That gives us an abstract cell decomposition for $\phi(\bar x; S)$. However, this decomposition does not need to be weakly definable.}
    \label{fig:atomic-acl}
\end{figure}

\begin{example}\label{ex:no-definition}
    Consider the structure $\M$ and the formula $\phi(x; y)$ from Example \ref{ex:crosses}.
    The obvious abstract cell decomposition for $\phi(x; y)$ mentioned above is the assignment $\T$ defined for every finite $S \subseteq M$ as follows
    \[
        \T(S) = \set{\set{v} : v \in S} \cup \set{M \setminus S}.
    \]
    However, $\phi(x; y)$ does not admit a weakly definable abstract cell decomposition in $\M$.
    Indeed, assume by contradiction that $\phi(x; y)$ admits an abstract cell decomposition $\T$ in $\M$ weakly definable by $\Psi(x; y_1, \ldots, y_k)$.
    Take a set $S$ of size $k + 1$ and an element $v \in M \setminus S$.
    By the definition of $\T$ there is a set $F \in \T(S)$ such that $v \in F$.
    This set $F$ is definable by some $\psi \in \Psi$, so $F = \psi(M; s_1, \ldots, s_k)$ for some $s_1, \ldots, s_k \in S$.
    Therefore, there is an element $s \in S \setminus \set{s_1, \ldots, s_k}$.
    As $\psi$ is a formula over the empty signature, it cannot distinguish $v$ from $s$, so $s \in F$.
    However, $\set{s}$ crosses $F$, which is a contradiction.
    \qed
\end{example}

\begin{example}\label{ex:definition}
    We observed in Example \ref{ex:no-definition} that the formula which is just an equality does not admit a weakly definable abstract cell decomposition in an infinite structure over the empty signature.
    We can overcome this problem by adding a total order to our structure.
    Formally, consider the structure $\M_\preceq$ over the signature $\Sigma = \set \preceq$ with the domain $M = \set{v_1, v_2, \ldots}$ and assume $v_i \preceq v_j \iff i \le j$.
    Consider an abstract cell decomposition $\T$ defined for any $S = \set{v_{i_1}, \ldots, v_{i_k}}$ with $i_1 < \ldots < i_k$ as follows
    \begin{equation*}
        \begin{split}
            \T(S) = & \set{\set{v} : v \in S} \\
            \cup & \set{\set{v : v_{i_j} \prec v \prec v_{i_{j + 1}}} : j \in [k - 1]} \\
            \cup & \set{\set{v : v \prec v_{i_1}}} \\
            \cup & \set{\set{v : v \succ v_{i_k}}}.
        \end{split}
    \end{equation*}
    Observe that every $F \in \T(S)$ can be defined by one of these four formulas:
    \begin{itemize}
        \item $\psi_1(x; y_1) \equiv x = y_1$;
        \item $\psi_2(x; y_1, y_2) \equiv y_1 \preceq x \land x \preceq y_2 \land x \neq y_1 \land x \neq y_2$;
        \item $\psi_3(x, y_1) \equiv x \preceq y_1 \land x \neq y_1$;
        \item $\psi_4(x; y_1) \equiv y_1 \preceq x \land x \neq y_1$.
    \end{itemize}
    Formally, these formulas have a different number of parameters, but we can extend $\psi_1, \psi_3$ and $\psi_4$ with one artificial parameter which is not further used and in this way we obtain that $\T$ is weakly definable with $\Psi(x; y_1, y_2)$ for $\Psi = \set{\psi_1, \psi_2, \psi_3, \psi_4}$.
    \qed
\end{example}

If we have an abstract cell decomposition $\T$ which is weakly definable by $\Psi(\bar x, \bar {y_1}, \ldots, \bar {y_k})$ we know that for every $F \in \T(S)$ there is a choice of a formula $\psi \in \Psi$ and $\bar{s_1}, \ldots, \bar{s_k} \in S$ such that $F = \psi(\M, \bar{s_1}, \ldots, \bar{s_k})$.
However, we don't know \emph{a priori} which choices of $\psi \in \Psi$ and parameters from $S$ lead to a set in $\T(S)$.
Therefore, we say that an abstract cell decomposition $\T$ for $\phi(\bar x; \bar y)$ is \emph{definable} if it is weakly definable by some $\Psi(\bar x; \bar {y_1}, \ldots, \bar {y_k})$ and if for every $S \subset M^{|\bar y|}$ and each $\Psi(\bar x; S^k)$-definable $\Delta \subseteq M^{|\bar x|}$ there is a set $\I(\Delta) \subseteq M^{|\bar y|}$, \emph{uniformly definable} in $\Delta$, such that
\[
    \T(S) = \set{\Delta \in \Psi(\bar x; S) : \I(\Delta) \cap S = \emptyset}.
\]
By the uniform definability of $\I(\Delta)$ we mean that for every $\psi(\bar x; \bar {y_1}, \ldots, \bar {y_k}) \in \Psi(\bar x; \bar {y_1}, \ldots, \bar {y_k})$, there is a formula $\theta_\psi(\bar y; \bar {y_1}, \ldots, \bar {y_k})$ such that for any $\bar {s_1}, \ldots, \bar {s_k} \in M^{|\bar y|}$ if $\Delta = \psi(M; \bar {s_1}, \ldots, \bar {s_k})$ then $\I(\Delta) = \theta_\psi(M; \bar {s_1}, \ldots, \bar {s_k})$.

Observe, that every abstract cell decomposition $\T$ weakly definable with $\Psi(\bar x; \bar {y_1}, \ldots, \bar {y_k})$ induces the definable abstract cell decomposition $\T^{\max}$ with
\[
    \I(\Delta) = \set{\bar s \in M^{|\bar y|} : \phi(\bar x; \bar s) \text{ crosses } \Delta},
\]
i.e. for every $\psi \in \Psi$ we take
\begin{dmath*}
    \theta_\psi(\bar y; \bar {y_1}, \ldots, \bar {y_k}) \equiv \left( \exists \bar x. \psi(\bar x; \bar {y_1}, \ldots, \bar {y_k}) \land \phi(\bar x; \bar y) \right) \land \left( \exists \bar x. \psi(\bar x; \bar {y_1}, \ldots, \bar {y_k}) \land \neg \phi(\bar x; \bar y) \right).
\end{dmath*}
Moreover, for any finite $S \in M^{|\bar x|}$ we have $|\T^{\max}(S)| \le |\Psi| \cdot |S|^k = O(|S|^k)$.

\begin{example}
    \label{ex:definable}
    In Example \ref{ex:definition} we gave an example of the abstract cell decomposition $\T$ weakly definable by $\Psi(x, y_1, y_2)$ for the formula $x = y$ in the model $\M_\preceq$.
    Consider now a set $S = \set{v_{1}, v_{3}, v_{5}}$.
    In $\T(S)$ there is a set $\set{v : v_{1} \prec v \prec v_{{3}}} = \set{v_2}$, which is defined by the formula $\psi_2(x; y_1, y_2) \equiv y_1 \preceq x \land x \preceq y_2 \land x \neq y_1 \land x \neq y_2$, where we plug $y_1 = v_1$ and $y_2 = v_3$.
    However, if we plug $y_1 = v_1$ and $y_2 = v_5$ we get the set $\Delta = \set{v : v_{1} \prec v \prec v_{{4}}} = \set{v_2, v_3, v_4} \not \in T(S)$.
    
    Consider the formula $\theta_{\psi_2}$ given above for defining $\T^{\max}$.
    In this particular case it says:
    \[
        \theta_{\psi_2}(y, y_1, y_2) \equiv \psi_2(y, y_1, y_2) \land (\exists x. x \neq y \land \psi_2(x, y_1, y_2)). 
    \]
    Observe that we have $\M_\preceq \models \theta_{\psi_2}(v_3, v_1, v_5)$, so we have $v_3 \in \I(\Delta)$.
    In this way we know that the set $\Delta$ shouldn't be used in $\T(S)$.
    It is easy to observe, that for the considered abstract cell decomposition $\T$, we have that $\T^{\max}(S) = O(|S|)$ for the induced $\T^{\max}$ definable abstract cell decomposition.
    \qed
\end{example}

Using the notion of a weakly definable abstract cell decomposition we can define the notion of distality.
We define it for single formulas and classes of structures similarly to the notion of NIP.
The following definition is an adaptation of \cite[Fact 2.9]{distal}.

\begin{definition}
    A formula $\phi(\bar x; \bar y)$ with two tuples of parameters $\bar x$ and $\bar y$ is distal in a class of structures $\C$ if there is a finite set of sentences $\Psi(\bar x; \bar y_1, \ldots, \bar y_k)$ such that for any $\M \in \C$ there is an abstract cell decomposition for $\phi(\bar x; \bar y)$ weakly definable by $\Psi$.
\end{definition}

In other words we say that a formula $\phi(\bar x; \bar y)$ is distal in $\C$ if in every $\M \in \C$ it admits an abstract cell decomposition weakly definable by the same set of formulas.
As we observed that every weakly definable abstract cell decomposition $\T$ induces a definable abstract cell decomposition $\T^{\max}$, we can assume that the decomposition in the definition of a distal formula is definable.
Moreover, we say that a class of structures $\C$ is distal if every formula $\phi(\bar x; \bar y)$ is distal in $\C$ and we say that a single structure $\M$ is distal if the class $\set{\M}$ is.

Observe, that similarly to the case of NIP classes of structures, if a formula $\psi(\bar x; \bar y)$ admits a distal cell decomposition weakly definable by a finite set of sentences $\Psi$ in $\C$, then it also admits a distal cell decomposition weakly definable by a finite set of sentences $\Psi$ in the elementary closure $\bar \C$.

\subsection{Distal combinatorial tools}
We start this section with the original statement of Distal cutting lemma.

\begin{theorem}[{Distal cutting lemma, {\cite[Theorem 3.2]{distal}}}]
    \label{theorem:cutting_lemma}
    Let $\M$ be a fixed structure and $\phi(\bar x; \bar y)$ be a formula admitting a definable distal cell decomposition $\T$ (weakly definable by a finite set of formulas $\Psi(\bar y; \bar y_1, \ldots, \bar y_k)$) with $\T(S) = O(|S|^d)$.
    Then for any $A \subseteq M^{|y|}$ of size $n$ and any real $r$ satisfying $1 \le r \le n$ we can partition $M^{|x|}$ into sets $\seq Xl$ with
    \[
        l \le Cr^{d},
    \]
    and satisfying the following: for every $X_i$ there are at most $\frac nr$ tuples $\bar a \in A$ such that $\phi(M; \bar a)$ crosses $X_i$.
    The constant $C$ depends only on the model $\M$ and the formula $\phi$.
    Moreover, each of the $X_i$'s is an intersection of at most two $\Psi(x; A^k)$-definable sets.
\end{theorem}

To better understand the statement of Theorem \ref{theorem:cutting_lemma} consider it in the case when $\M$ is a graph (possibly ordered and infinite) and $\phi(x; y) \equiv E(x, y)$. 
Then Theorem \ref{theorem:cutting_lemma} says, that for any set of vertices $A$ we can partition all the vertices into at most $Cr^d$ sets $\seq Xl$ and in every $X_i$ there are at most $\frac nr$ vertices $a \in A$ such that $a$ has both a neighbor and a non-neighbor in $X_i$, so in other words $X_i$ is homogenous towards all but $\frac nr$ vertices from $A$.

By using compactness in a standard way, we get a version of Theorem \ref{theorem:cutting_lemma} for classes of structures.

\begin{theorem}[Distal cutting lemma for classes of structures]
    \label{theorem:finitary_cutting_lemma}
    Let $\C$ be a fixed class of structures and $\phi(\bar x; \bar y)$ be a formula admitting for every $\M \in \C$ a definable distal cell decomposition $\T_\M$ (weakly definable by the same finite set of formulas $\Psi(\bar y; \bar y_1, \ldots, \bar y_k)$) with $\T_\M(S) = O(|S|^d)$.
    Then for any $\M \in \C$, any $A \subseteq M^{|y|}$ of size $n$, and any real $r$ satisfying $1 \le r \le n$ we can partition $M^{|x|}$ into sets $\seq Xl$ with
    \[
        l \le Cr^{d},
    \]
    such that for every $X_i$ there are at most $\frac nr$ tuples $\bar a \in A$ such that $\phi(M; \bar a)$ crosses $X_i$.
    The constant $C$ depends only on the class of structures $\C$ and the formula $\phi$.
    Moreover, each of the $X_i$'s is an intersection of at most two $\Psi(x; A^k)$-definable sets.
\end{theorem}
\begin{proof}
    Assume by contradiction that the statement is not true for some real $r$.
    Therefore, for every $L \in \N$ there is some $n_L \in \N$ and a structure $\M_L \in \C$ together with $A_L \subseteq M_L^{|y|}$ of cardinality at most $n_L$ such that we cannot choose at most $Lr^d$ sets which are intersection of at most two $\Psi(x; A_L^k)$-definable sets and form a desired partition.
    Without loss of generality, we can choose constants $n_L$ so that they form a non-decreasing sequence.
    
    Observe, that for every $L$ we can write a sentence $\xi_L$ of first-order logic stipulating:
    \begin{center}
        "there is a set $A_L^k$ of at most $n_L$ tuples of length $|y|$ (formally we quantify existentially $n_L$ times over tuples of size $|y|$) with the following property: for every choice of $Lr^d$ sets which are intersection of at most two $\Psi(\bar x; A_L^k)$-definable sets, if they cover the whole domain, then at least one of them is crossed by more than $\frac {|A_L^k|}r$ $\phi(\bar x; A_L^k)$-definable sets". 
    \end{center}
    
    Consider the theory $T$ of the class $\C$.
    Extend it by adding all the sentences $\xi_L$ to obtain a new theory $T'$.
    By compactness, $T'$ is consistent.
    Indeed, any finite subset of $T'$ is consistent, as there is a model in $\C$ which satisfies sentences $\xi_L$ for arbitrary large $L$'s.
    
    As $T'$ is consistent, it has a model $\M'$.
    As $T \subseteq T'$, the model $\M'$ satisfies the assumptions of Theorem \ref{theorem:cutting_lemma}.
    However, as $\xi_L \in T'$ for every $L \in \N$, then $\M'$ does not satisfy its conclusion, which gives us a contradiction.
\end{proof}

A similar approach can be applied to another distal combinatorial tool - Distal regularity lemma \cite[Theorem 5.8]{regularity}.
However, its full statement is more complicated and uses a number of new definitions.
Moreover, it is stated for distal models, but a careful analysis of the proof shows that it can be applied also on the level of distal formulas provided that the model we are working with is NIP.
Therefore, we decided to present here a simplified statement of this theorem in the setting of (ordered) graphs.

\begin{theorem}[Simplified version of Distal regularity lemma]
    \label{theorem:simple_regularity_lemma}
    Let $\C$ be an NIP class of (ordered) graphs and assume that the formula $\phi(x; y) \equiv E(x, y)$ is distal in $\C$. 
    Then, there is a constant $c = c(\C)$ such that for any $\epsilon > 0$ and for any graph $G \in \C$ there is a partition $\seq Pl$ of vertices of $G$ for some $l \le c(\frac 1\epsilon)^{c}$, such that
    \[
        \sum |P_i||P_j| \le \epsilon |V(G)|^2
    \]
    where the sum is over all $i, j \in [|V(G)|]$ such that $P_i$ and $P_j$ are not homogenous.
\end{theorem}

A sketch of the proof of Theorem \ref{theorem:simple_regularity_lemma} is given in 
\iflipics
Appendix \ref{appendix:proof-regularity-lemma}.
\else
the full version of this paper.
\fi
As we mentioned, this is just repeating parts of the argument from \cite{regularity}, but in the simplified setting of graphs.

\section{Abstract cell decomposition for the edge relation in graphs of bounded twin-width}
\label{chap:encoding}

Classes of ordered graphs of bounded twin-width are known to be distal.
The proof of this fact is a combination of a few involved results from model theory.
Without going into details, if $\widehat{\C}$ is a class of totally ordered graphs of bounded twin-width, then by \cite[Theorem 1 and Theorem 3]{twinwidth4} it is monadically NIP (this is a property of models much stronger than being NIP).
Then by \cite[Theorem 1.1]{bl21} we get that since $\widehat{\C}$ is monadically NIP then it is dp-minimal (this is another property of models which we don't define in this paper), and finally by \cite[Lemma 2.10]{distality-definition} we get that as $\widehat{\C}$ is dp-minimal and totally ordered, then it is distal.
However, this proof does not give any explicit abstract cell decomposition for the class $\C$, which means that e.g. we don't have any bound on the parameter $d$ in Theorem \ref{theorem:finitary_cutting_lemma}.

The aim of this section and the main result of this paper is a direct combinatorial proof of the result sketched above for the case of the edge relation.
Formally, we prove the following theorem.

\begin{theorem}
    \label{theorem:encoding_theorem}
    Let $\widehat{\C}$ be a class of ordered graphs with the following property: there is a constant $t \in \N$ depending only on $\widehat{\C}$ such that for every $(G, \preceq) \in \widehat{\C}$ the matrix $M_\preceq(G)$ is $t$-mixed free.
    Then the formula $\phi(x; y) \equiv E(x, y)$ admits in $\widehat{C}$ a distal cell decomposition weakly definable by a finite set $\Psi(x; y_1, \ldots, y_k)$, where $k = \mathcal O(t)$.
\end{theorem}

\subsection{Proof of Theorem \ref{theorem:encoding_theorem}}

We start by introducing a crucial definition for this part.

\begin{definition}
Let $M$ be an $m \times n$ $0-1$ matrix.
We define the \textbf{corner-profile} of column $j$ (with $1 \le j < n$) as the subset $\prof(j) \subseteq [m - 1]$ such that $i \in \prof(j)$ if and only if $M[i:(i+1)][j:(j+1)]$ is a corner.
For a subset of columns $J \subseteq [n - 1]$ we define $\prof(J) = \bigcup_{j \in J} \prof(j)$.
For an ordered graph $(G, \preceq)$ we identify columns of $M_\preceq(G)$ with vertices of $G$ and therefore assume that the domain of $\prof$ is $V(G)$ and it ranges over subsets of $V(G)$.
\end{definition}

An example of a corner profile of a column is presented in Figure \ref{fig:corner-profile}.

\begin{figure}
    \centering
    \includegraphics[scale=0.5,page=3]{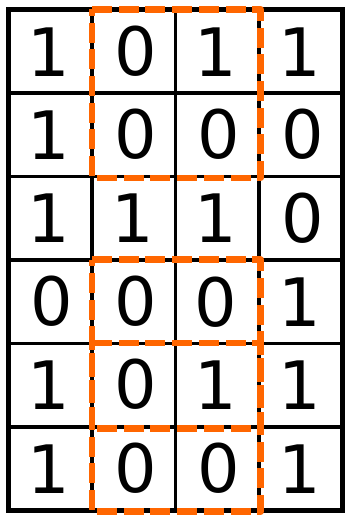}
    \caption{In this example we have $\prof(2) = \{1, 4, 5\}$.}
    \label{fig:corner-profile}
\end{figure}

Before we proceed with the proof of Theorem \ref{theorem:encoding_theorem}, it is insightful to prove a simpler proposition, which shows some intuitions that will be useful later on.
Namely, we prove that vertices with a \textit{large} corner-profile can be efficiently defined using a few parameters.

We use additional notation for an ordered graph ${(G, \preceq)}$.
A vertex $v'$ is the successor of vertex $v$ in $\preceq$ (denoted by $s(v)$) if $v \neq v', v \preceq v'$ and for all $u \in V(G)$ we have either $u \preceq v $ or $v' \preceq u$.
Two vertices $v, v' \in V(G)$ are adjacent in $\preceq$ if $v' = s(v)$ or $v = s(v')$.
For a vertex $v$ we also denote by $S^k(v)$ the set of $k$ consecutive vertices after $v$ together with $v$, i.e. $S^k(v) = \set{s^l(v) : 0 \le l \le k}$.

\begin{proposition}
\label{prop:one_vertex_large_profile}
Let $(G, \preceq)$ be a $t$-mixed free ordered graph and let $v \in V(G)$ satisfy $|\prof(v)| \ge 2t$.
There are a constant $f_t$ depending only on $t$, a sequence of vertices $(\seq at) \in \prof(v)^t$, and a first-order formula in the language of ordered graphs $\phi(x; \seq yt)$ such that the length of $\phi$ is bounded by $f_t$ and $\phi(G; a_1, \ldots, a_t) = \set{v}$.
\end{proposition}
\begin{proof}
As $|\prof(v)| \ge 2t$ we can pick a subset $P \subseteq \prof(v)$ of size $t$ such that every two vertices $w, w' \in P$ are not adjacent in $\preceq$.
Let us denote $S = \{u \in V(G): P \subseteq \prof(u)\}$.
Obviously $v \in S$ and we can easily see that $|S| < 2t$.
Indeed, if $|S| \ge 2t$, then we could pick $S' \subseteq S$ of size $t$ such that every two vertices in $S'$ are not adjacent in $\preceq$.
Then pairs of rows $\set{(w, s(w)) : w \in P}$ and pairs of columns $\set{(u, s(u)) : u \in S'}$ would induce a $t$-mixed minor of $M_\preceq(G)$.
That is a contradiction with the assumption that $(G, \preceq)$ is $t$-mixed free.
Therefore, we take $\phi$ which just says that $v$ is the $l$'th vertex in order $\preceq$ such that $P \subseteq \prof(v)$.
This can be done within the constraints in the statement, as we can use vertices in $P$ as $a_1, \ldots, a_t$ and we can define the successor of a given vertex using $\preceq$.
Moreover, the size of $\phi$ depends only on $t$ (as we discussed that $l$ is smaller than $2t$), so it can be bounded by some constant~$f_t$.
\end{proof}

The assumption that one vertex has a large corner-profile is quite strong, in the sense that we can easily imagine a matrix which admits a large mixed minor, but its every column has a small corner-profile (e.g. of a constant size).
However, Proposition \ref{prop:one_vertex_large_profile} can be leveraged to the case when a set of consecutive vertices has a large corner-profile, which is formalized in Proposition \ref{prop:l_vertex_large_profile}.

\begin{proposition}
\label{prop:l_vertex_large_profile}
Let $(G, \preceq)$ be a $t$-mixed free ordered graph and let $v \in V(G)$ satisfy $|\prof(S^l(v))| \ge 2t$ for some integer $l$.
There are a constant $f_t$ depending only on $t$, a sequence of vertices $(\seq at) \in \prof(S^l(v))^t$, a first-order formula in the language of ordered graphs $\phi(x; \seq yt)$, and a vertex $v' \in S^l(v)$ such that the length of $\phi$ is bounded by $f_t$ and $\phi(G; a_1, \ldots, a_t) = \set{v'}$.
\end{proposition}

\begin{figure}
    \centering
    \iflipics
    \includegraphics[scale=0.8,page=5]{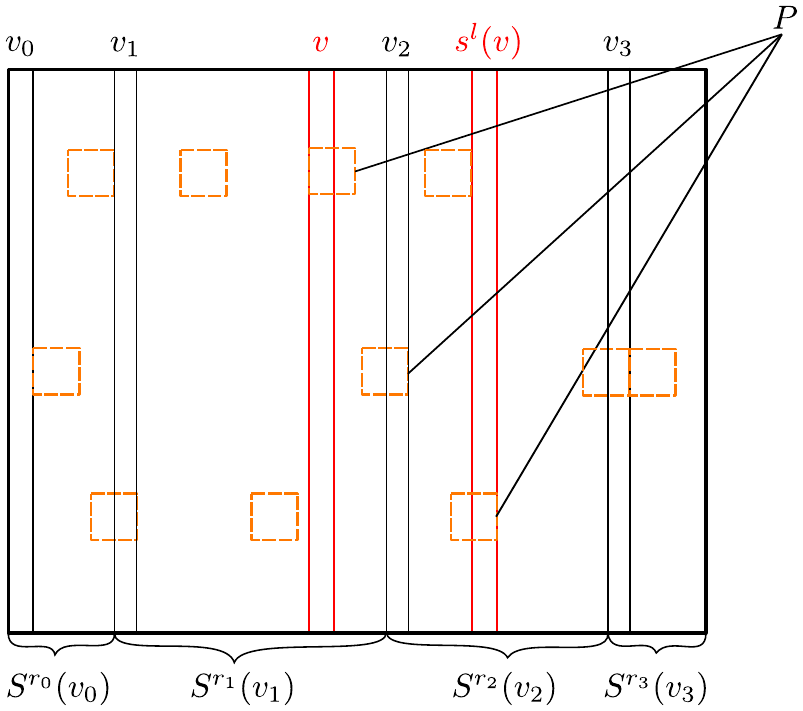}
    \else
    \includegraphics[scale=0.6,page=5]{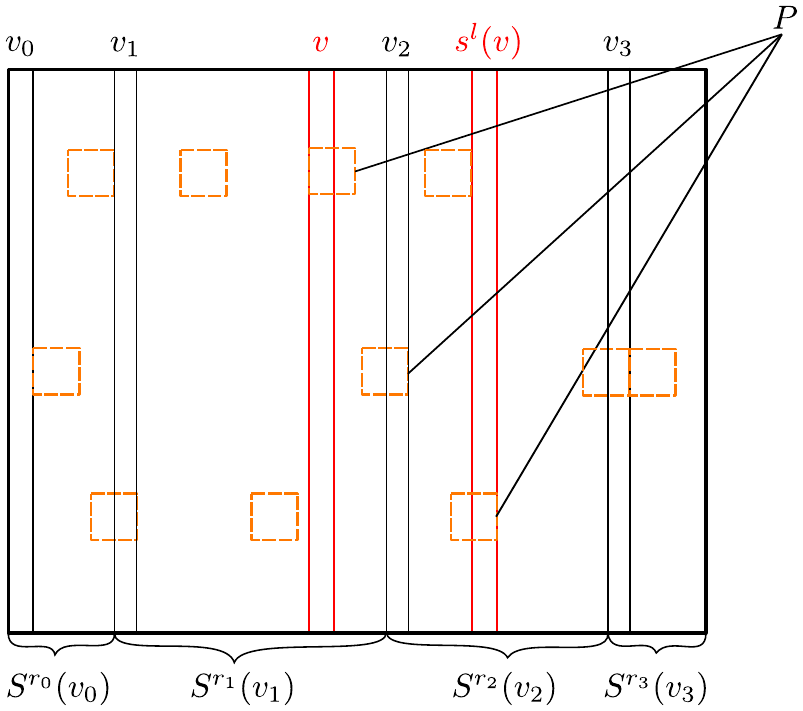}
    \fi
    \caption{Example of how the $v_i$ are set in the proof of Proposition \ref{prop:l_vertex_large_profile}.}
    \label{fig:l_vertex_prop_fig}
\end{figure}

\begin{proof}
Similarly as in the proof of Proposition \ref{prop:one_vertex_large_profile} we can pick a subset $P \subseteq \prof(S^l(v))$ of size $t$ such that every two vertices $w, w' \in P$ are not adjacent in $\preceq$.
Let $v_0$ be the smallest vertex of $G$ in $\preceq$ and $r_0 = \min \set{n \in \N: P \subseteq \prof(S^n(v_0))}$, i.e. $r_0$ is the smallest integer such that first $r_0 + 1$ vertices of $G$ in $\preceq$ have $P$ as a subset of their corner-profile.
This is obviously well-defined as $P \subseteq S^l(v)$.
Now we set $v_1 = s^{r_0 + 1}(v_0)$ and $r_1 = \min \set{n \in \N: P \subseteq \prof(S^n(v_1))}$.
We define $v_2, r_2, v_3, \ldots$ analogously, until some $v_j$ or $r_j$ is not defined (i.e. $s^{r_{j - 1}}(v_{j - 1})$ is the maximal vertex of $G$ in $\preceq$ or $\set{n \in \N : P \subseteq \prof(S^n(v_j))}$ is empty).
We clearly have $j < 2t$.
Indeed, if $j \ge 2t$ then the sets of columns $S^{r_0 + 1}(v_0), S^{r_2 + 1}(v_2), \ldots, S^{r_{2t - 2} + 1}(v_{2t - 2})$ are pairwise disjoint.
Moreover, together with pairs of rows $\set{(u, s(u)) : u \in P}$ they induce a $t$-mixed minor of $M_\preceq(G)$, which is a contradiction.

If $v_{j - 1}, r_{j - 1}$ is the last pair, which was well-defined above, and we also managed to define $v_j$ but not $r_j$, then we can't have $v_j \preceq v$.
That is because we have $P \subseteq \prof(S^l(v))$, so in this case $r_j$ would be well-defined.
Therefore, there is some $p$ such that $v \in S^{r_p}(v_p)$. 
We also have $s^{r_p}(v_p) \in S^l(v)$.
Indeed, we can't have $s^{r_p}(v_p) \prec v$ by the definition of $p$ and if we had $s^l(v) \prec s^{r_p}(v_p)$, then we would have $P \subseteq S^l(v) \subseteq S^{r_p - 1}(v_p)$, which contradicts the definition of $r_p$.
We pick the vertex $s^{r_p}(v_p)$ as our $v'$ from the statement.
That is because it is either the maximal vertex of $G$ (in which case it is obvious that we can define it even by a formula without parameters) or the predecessor of $v_{p + 1}$.
However, $v_{p + 1}$ is definable by a first-order formula of ordered graphs with $\seq at$ being the vertices of $P$ in some order.
This is because the definition of each $v_i$ can be easily transformed into a first-order formula of ordered graphs, and the size of formula for $v_p$ can be bounded by some constant $f_t$ because $p \le 2t$.
Finally, as we mentioned in the proof of Proposition \ref{prop:one_vertex_large_profile}, the successor relation is definable, so indeed we can define $s^{r_p}(v_p)$ as the predecessor of $v_{p + 1}$.
\end{proof}

Let us remark that Proposition \ref{prop:l_vertex_large_profile} assures that if we have a set of consecutive vertices with a large corner-profile, then we can encode \textit{some} vertex in it using a first-order formula.
However, we cannot control which vertex from the given set is to be encoded.
It turns out not to be that problematic, so we can start the proof of Theorem \ref{theorem:encoding_theorem}.

\begin{proof}[Proof of Theorem \ref{theorem:encoding_theorem}]
Let us take a $t$-mixed free ordered graph $(G, \preceq)$ and some non-empty $A \subseteq V(G)$.
Our goal is to present an abstract cell decomposition which is weakly definable by the same set of formulas which does not depend on $G$ nor on $A$.

Similarly as in the proof of Theorem \ref{theorem:neighborhood-complexity-theorem} we can remove from $M_\preceq(G)$ rows that correspond to vertices in $V(G)\setminus A$, thus obtaining a $t$-mixed free $|A| \times |V(G)|$ matrix $M$.
We set $v_0$ to be the minimal vertex of $G$ in $\preceq$ and $r_0 = \min \set{n \in \N: |\prof(S^n(v_0))| \ge 2t}$.
Next, we set $v_1 = s^{r_0 + 1}(v_0)$ and $r_1 = \min \set{n \in \N: |\prof(S^n(v_1))| \ge 2t}$.
We proceed this way with $v_2, r_2, v_3, \ldots$. If at some point we manage to define $v_l$ but the set $\set{n \in \N: |\prof(S^n(v_l))| \ge {2t}}$ is empty, then we define $r_l = \max \set{n \in \N: s^n(v_l) \in V(G)}$, i.e. $S^{r_l}(v_l) = \set{u \in V(G): v_l \preceq u}$.
In this way we define $l + 1$ pairs $v_0, r_0, \ldots, v_l, r_l$.
This process is presented in Figure \ref{fig:l_vertex_prop_fig}.

We have that in every $S^{r_j}(v_j)$ there is some $v_j'$ which is first-order definable.
This is a bit more tricky than in the proof of Proposition \ref{prop:l_vertex_large_profile}, as now any two consecutive rows of $M$ don't necessarily correspond to vertices of $V(G)$ which are adjacent in $\preceq$.
However, we are still able to define a corner by explicitly naming both its rows.
Therefore, the formula $\phi_j$ which defines $v_j'$ needs to use $2t$ parameters which correspond to $t$ vertices from $\prof(S^{r_j}(v_j))$ and their successors in $\preceq_{\restriction {A \times A}}$.
There is also a special case for $S^{r_l}(v_l)$, as we can't assume $|\prof(S^{r_l}(v_l))| \ge 2t$.
Nevertheless, the vertex $s^{r_l}(v_l)$ is definable as the maximal vertex of $G$, so we can set $v_l' = s^{r_l}(v_l)$.
As $M$ is $t$-mixed free, then the size of each $\phi_j$ is bounded by some function of $t$ only.

If for any $s^{r_j}(v_j)$ we have $\prof(s^{r_j}(v_j)) \ge 2t$, then we can define $s^{r_j}(v_j)$ and $v_{j+1}$ (if it exists) as well.
Again, we need to use for it formulas with $2t$ parameters of size bounded by some function of $t$.

By now we defined a few singleton sets with a bounded number of formulas such that any of them uses at most $2t$ parameters from $A$.
Now we will show how to deal with the remaining vertices.
Every such vertex $u$ satisfies $w \prec u \prec w'$ for some $w, w'$ already defined in previous steps.
Let us denote $I_{ww'} = \set{u : w \prec u \prec w'}$.
From the construction we have $|\prof(I_{ww'})| < 4t$.
Moreover, by Corollary \ref{cor:representative_set} we know that there is a subset $A_{ww'} \subseteq A$ of size at most $|\prof(I_{ww'})| + 1 \le 4t$ such that whenever two vertices $u, u' \in I_{ww'}$ have exactly the same neighborhood on the set $A_{ww'}$, then their columns in $M$ are equal, so they have the same neighborhood in the whole $A$.
Therefore, we can define sets of vertices between $w$ and $w'$ in $\preceq$, that have specific values on $A_{ww'}$ using formulas with $2t+2t+4t$ parameters and of size bounded by some function of $t$.
Finally, we obtain the statement of Theorem \ref{theorem:encoding_theorem} with $k = 8t$.
The set of formulas $\Psi$ that we use is finite, because every $\psi \in \Psi$ has its size bounded by a function of $t$ only.
\end{proof}

\subsection{Distal tools for graphs of bounded twin-width}

We already know that for every graph $G$ from a class $\C$ of unordered graphs of twin-width bounded by $t$ we can find a total order $\preceq$ on its vertices, such that $(G, \preceq)$ is $2t+2$-mixed free.
Combining this fact with Theorem \ref{theorem:encoding_theorem} and Theorem \ref{theorem:finitary_cutting_lemma} we obtain cutting lemma for graphs of bounded twin-width (without an order).

\begin{theorem}[Cutting lemma for graphs of bounded twin-width]
    \label{theorem:cutting-lemma-tww}
    Let $G$ be a graph of twin-width at most $t$.
    There is a constant $C$ depending only on $t$ such that for any $A \subseteq V(G)$ of size $n$ and any real $r$ satisfying $1 \le r \le n$ we can partition $V(G)$ into sets $X_1, \ldots, X_l$ with
    \[
        l \le Cr^{16t + 16}
    \]
    such that for every $X_i$ all but at most $\frac nr$ vertices $a \in A$ are homogenous towards $X_i$.
\end{theorem}

Similarly, using Theorem \ref{theorem:simple_regularity_lemma} we can formulate regularity lemma for graphs of bounded twin-width.
In this case, we cannot directly apply the bounds on the size of abstract cell decompositions from Theorem \ref{theorem:encoding_theorem}.
The reason is that the bounds in Theorem \ref{theorem:simple_regularity_lemma} depend also on VC-dimension of set systems definable by formulas used in a definition of an abstract cell decomposition.

\begin{theorem}[Regularity lemma for graphs of bounded twin-width]
    \label{theorem:regularity-lemma-tww}
    Let $G$ be a graph of twin-width at most $t$.
    Then, there is a constant $c = c(t)$ such that for any $\epsilon > 0$ there is a partition $P_1, \ldots, P_l$ of $V(G)$ for some $l \le c (\frac 1\epsilon)^{c}$, such that
    \[
        \sum |P_i||P_j| \le \epsilon |V(G)|^2
    \]
    where the sum is over all $i, j \in [|V(G)|]$ such that $P_i$ and $P_j$ are not homogenous.
\end{theorem}

\section{Conclusions}
\label{chap:conclusions}

One of the most important conjectures in structural and algorithmic graph theory is that model checking is FPT on a class of graphs if and only if the class is monadically NIP \cite{conjecture-nip}.
As we mentioned, this notion comes from model theory and therefore it is believed that progressing towards this conjecture requires a better understanding of the interplay of model theory and structural graph theory.
We believe that studying the notion of distality is a step in this direction.

We can observe the mentioned interplay when we can prove a theorem in two completely different ways -- one using tools from model theory, and the other being a direct, combinatorial proof.
An example of such theorem is shown in \cite{flipper}.
We are grateful to the Anonymous Reviewer who pointed out that Theorem \ref{theorem:regularity-lemma-tww} alternatively can be proven in a purely combinatorial way using \cite[Lemma 20]{twinwidth2}.
We believe that studying connections between these two proofs is an interesting goal.

As we mentioned in Section \ref{chap:encoding} a class of ordered graphs is distal if and only if it has bounded twin-width.
However, the distal regularity lemma and the distal cutting lemma does not only apply to graphs of bounded twin-width.
As an example, we can consider classes of graphs of bounded expansion equipped with an arbitrary order.
They are not distal, but every first-order formula which doesn't use order is distal in them \cite{distal-be}.
This leads to a possible generalization of classes of graphs of bounded expansion and bounded twin-width and such a generalization is highly anticipated and sought.
In light of these facts, Theorem \ref{theorem:cutting-lemma-tww} and Theorem \ref{theorem:regularity-lemma-tww} may be seen as a step towards finding common combinatorial properties of the two notions.

\bibliography{bibliography}

\appendix
\section{Proof of Lemma \ref{lemma:nip-twin-width}}\label{appendix:twin-width-mon-nip}

In Section \ref{chapter:preliminaries} for a given graph $G$ and a total ordering $\sigma$ of $V(G)$ we denoted by $M_\sigma(G)$ the adjacency matrix of $G$ in the order $\sigma$.
This definition can be lifted to adjacency matrices of arbitrary binary structures, in the following way, which was presented in \cite{twinwidth1}.

Let $H$ be a binary structure over the signature consisting of $k$ binary relations $R_1, \ldots, R_k$ and let $\sigma$ be a total ordering of the domain of $H$, say $v_1, \ldots, v_n$.
For vertices $v_i, v_j$ and a relation $R_l$ we define
\[
    e_{ij}^l = \begin{cases}
        0 &\text{ if } \neg R_l(v_i, v_j) \land \neg R_l(v_j, v_i) \\
        1 &\text{ if } R_l(v_i, v_j) \land \neg R_l(v_j, v_i) \\
        -1 &\text{ if } \neg R_l(v_i, v_j) \land R_l(v_j, v_i) \\
        2 &\text{ if } R_l(v_i, v_j) \land R_l(v_j, v_i),
    \end{cases}
\]
and $e_{ij} = (e_{ij}^1, \ldots, e_{ij}^k)$.
Finally, we define the adjacency matrix of $H$ in the order $\sigma$ as the matrix $M_\sigma(H)$ satisfying $M_\sigma(H)[i][j] = e_{ij}$.

Not only the definition of $M_\sigma$ can be lifted to the case of binary structures, but also the definition of twin-width.
We don't present the lifted definition here.
Instead, we just remark that by \cite[Theorem 14]{twinwidth1} if for a given binary structure $H$ there is a total order $\sigma$ on its domain such that $M_\sigma(H)$ is $t$-mixed free, then the twin-width of $H$ is $2^{2^{\mathcal O(t)}}$.
Therefore, we can show that a binary structure has bounded twin-width, by finding an ordering in which its adjacency matrix is $t$-mixed free for some $t$.
This happens to be the case for the ordered graphs that we consider in the statement of Lemma \ref{lemma:nip-twin-width}.

\begin{lemma}
    \label{lemma:nip-aux}
    Let $(G, \preceq)$ be an ordered graph such that the matrix $M_\preceq(G)$ is $t$-mixed free.
    Then $M_\preceq((G, \preceq))$ is $2t$-mixed free.
\end{lemma}
Before starting the proof, let us stress the difference between $M_\preceq(G)$ and $M_\preceq((G, \preceq))$.
The first one is the adjacency matrix of the graph $G$ in the order $\preceq$.
In particular, it is a $0-1$ matrix.
On the other hand, $M_\preceq((G, \preceq))$ is an adjacency matrix of a binary structure over the signature $\set{E, \preceq}$, so its elements are from the set $(0, -1, 1, 2)^2$.
However, as the order $\preceq$ is both a part of the structure and the order which is used for creating $M_\preceq((G, \preceq))$, it is easier to reason about properties of this matrix.
\begin{proof}[Proof of Lemma \ref{lemma:nip-aux}]
    Let us assume by contradiction, that $M_\preceq((G, \preceq))$ admits a $(2t, 2t)$-division $(\Row, \Col) = (\set{\seq R{2t}},\linebreak[1] \set{\seq C{2t}})$, which is a $2t$-mixed minor.
    Denote by $R_{ij}$ (respectively $C_{ij}$) the sum $R_{ij} = \cup_{k = i}^jR_k$ (respectively $C_{ij} = \cup_{k = i}^jC_k$).
    It is easy to see that on one of the zones $R_{1t} \cap C_{t+1,2t}$ or $R_{t+1, 2t} \cap C_{1t}$ the relation $\preceq$ is constant.
    Therefore, either $(\set{R_1, \ldots, R_t}, \set{C_{t+1}, \ldots, C_{2t}})$ or $(\set{R_{t+1}, \ldots, R_{2t}}, \set{C_{1}, \ldots, C_{t}})$ induce a $t$-mixed minor of $M_\preceq(G)$, which is a contradiction.
\end{proof}

The notion of NIP classes of structures can be extended to much stronger notion of monadically NIP classes.
Roughly speaking a class $\C$ is monadically NIP if whenever we add a number of unary predicates to every structure in $\C$ thus obtaining a class $\bar{\C}$, then $\bar{\C}$ is NIP.
We don't define the notion of monadically NIP classes of structures formally, because for the proof of Lemma \ref{lemma:nip-twin-width} it is enough to know that this is a more restrictive notion than NIP.
Knowing that we can proceed to the proof of Lemma \ref{lemma:nip-twin-width}.

\begin{proof}[Proof of Lemma \ref{lemma:nip-twin-width}]
    By Lemma \ref{lemma:nip-aux} and \cite[Theorem 14]{twinwidth1} we know that the class $\hat{\C}$ of ordered graphs has bounded twin-width.
    By \cite[Theorem 11]{twinwidth4} and the equivalence of conditions $(1)$ and $(2)$ of \cite[Theorem 3]{twinwidth4} we get that $\hat{\C}$ is monadically NIP and hence NIP.
\end{proof}

\section{Sketch of the proof of Theorem \ref{theorem:simple_regularity_lemma}}
\label{appendix:proof-regularity-lemma}

In this section whenever we talk about graphs we assume that they can be ordered.
Before we sketch the proof of Theorem \ref{theorem:simple_regularity_lemma} we need to define a few auxiliary notions.

Let us assume that we have a class of graphs $\C$ such that the formula $\phi(x; y) \equiv E(x, y)$ admits an abstract cell decomposition in $\C$ weakly definable by $\Psi(x; y_1, \ldots, y_l)$.
For any $G \in \C$, any $\bar d \in G^l$, and any $\psi \in \Psi$ we denote by $C_{G, \psi, \bar d}$ the subset of vertices of $G$ defined by $\psi(G; \bar d)$ and call it a \emph{chamber}.
If in addition $\bar d \in B^l$ for some $B \subseteq V(G)$ then we say that $C_{G, \psi, \bar d}$ is a \emph{$B$-definable} chamber.

For a chamber $C$ and a set $B$ we denote by $C^\#(B)$ the set of all $b \in B$ such that $C$ and $\set{b}$ are not homogenous.
Note that $C^\#(G)$ is a definable set (by a formula depending just on the formula $\psi$ defining $C$).

For $B \subseteq V(G)$ we say that a chamber $C$ is \emph{$B$-complete} if $C$ is $B$-definable and $C^\#(B) = \emptyset$.

Since $\Psi$ weakly defines an abstract cell decomposition of $\phi$ in $\C$ then for any $G \in \C$, any finite $B \subseteq V(G)$, and any $a \in V(G)$ there is a $B$-complete chamber $C$ such that $a \in C$.
In particular for any finite $B \subseteq V(G)$ the union of all $B$-complete chambers covers $V(G)$.

For a positive $r \in \R$ and a $G \in \C$ we say that a family $\F$ of chambers is a \emph{$1/r$-cutting} if $V(G)$ is covered by $\F$ and for every $C \in \F$ we have $C^\#(G) \le |V(G)|/r$.

We also need the following fact, which can be seen as a special case of the VC-theorem.
\begin{lemma}[Simplified version of {\cite[Fact 2.2]{regularity}}]
    \label{lemma:vc-theorem}
    For any $k>0$ and $\epsilon > 0$ there is $n = \mathcal O(k (\frac 1\epsilon)^2 \log \frac 2\epsilon)$ satisfying the following. For any finite set $X$ and a family $\F$ of subsets of $X$ of VC-dimension $\le k$ there are some $x_1, \ldots, x_n \in X$ such that for any $S \in \F$ if $|S| \ge \epsilon |X|$ then $S \cap \set{x_1, \ldots, x_n} \neq \emptyset$.
\end{lemma}

After explaining these additional notions we can prove Lemma \ref{lemma:s-complete-chambers}, which is an important step in the direction of Theorem \ref{theorem:simple_regularity_lemma}.

\begin{lemma}[compare with {\cite[Claim 3.5]{regularity}}]
    \label{lemma:s-complete-chambers}
    There is a constant $K$ depending only on $\C$ such that the following holds. For any positive $r$ and for any $G \in \C$ there is a set $S \subseteq V(G)$ of size at most $Kr^2\log 2r$ such that the family of all S-complete chambers is a 1/r-cutting.
\end{lemma}
\begin{proof}
    Consider the family of sets
    \[
        \mathcal D = \set{C^\#(G): C \text{ is a } \text{$G$-definable chamber}}.
    \]
    This is a definable family, so it has VC-dimension bounded by a constant that depends only on $\C$.
    By applying Lemma \ref{lemma:vc-theorem} with $\epsilon = \frac 1r$ we obtain a subset $S \subseteq G$ of size at most $Kr^2\log 2r$ such that for every $G$-definable chamber $C$ if $|C^\#(G)| \ge |V(G)|/r$ then $S \cap C^\#(G) \neq \emptyset$. Since $C^\#(S) = \emptyset$ for any $S$-complete chamber $C$, we are done.
\end{proof}
By following the proof of \cite[Theorem 3.6]{regularity} and using Lemma \ref{lemma:s-complete-chambers} we can prove the following theorem.
\begin{theorem}
    \label{theorem:erdos-hajnal}
    Let $\C$ be an NIP class of graph such that $\phi(x; y) \equiv E(x, y)$ is distal in $\C$.
    Then, there is a constant $\delta$ depending only on $\C$ and a pair of formulas $\psi_1(x, \bar z_1), \psi_2(y, \bar z_2)$ such that for every $G \in \C$ there are $\bar c_1 \in V(G)^{|z_1|}, \bar c_2 \in V(G)^{|z_2|}$ such that $|\psi_1(G, \bar c_1)|, |\psi_2(G, \bar c_2)| \ge \delta |V(G)|$ and the pair of sets $\psi_1(G, \bar c_1), \psi_2(G, \bar c_2)$ is homogenous.
\end{theorem}

Now we start the sketch of the proof of Theorem \ref{theorem:simple_regularity_lemma}.
The idea is to consider for any $G \in \C$ rectangular partitions of $V(G)^2$, i.e. partitions where every part is of the form $S_1 \times S_2$ for $S_1, S_2 \subseteq V(G)$.
In particular, we consider partitions that are $(\psi_1, \psi_2)$-definable over some $A \subseteq V(G)$, for some formulas $\psi_1(x, \bar z_1), \psi_2(y, \bar z_2)$.
It means that for every $S_1 \times S_2$ in our partition $S_1$ (respectively $S_2$) can be expressed as a finite Boolean combination of sets from the family $\set{\psi_1(V(G), \bar a_1): \bar a_1 \in A^{|z_1|}}$ (respectively $\set{\psi_2(V(G), \bar a_2): \bar a_2 \in A^{|z_2|}}$).
For a rectangular partition $\Pp$ of $V(G)^2$ we can define its defect as:
\[
    \text{def}(\Pp) = \sum_{\substack{S_1 \times S_2 \in \Pp \\ S_1 \text{ and } S_2 \text{ not homogenous}}} |S_1||S_2|.
\]

The idea behind the proof of Theorem \ref{theorem:simple_regularity_lemma} is to start with the rectangular partition consisting of just one part $V(G) \times V(G)$ and then iteratively decrease its defect using Theorem \ref{theorem:erdos-hajnal} (possibly applied to subgraphs of $G$).
In this way we obtain a rectangular partition of $V(G)^2$ with a low defect which is $(\psi_1, \psi_2)$-definable over some set $A$ of size polynomial in $1/\epsilon$.
This process is depicted in \cite[Claim 5.7]{regularity}.
We finish the proof by observing, that the statement of Theorem \ref{theorem:simple_regularity_lemma} follows if we take the partition of $V(G)$ into $(\psi_1, \psi_2)$-types over $A$, which is formally stated in the proof of \cite[Theorem 5.8]{regularity}.


\end{document}

\else

\documentclass[conference]{IEEEtran}

\usepackage{cite}
\usepackage{amsmath,amssymb,amsfonts}
\usepackage{algorithmic}
\usepackage{graphicx}
\usepackage{textcomp}
\usepackage{xcolor}

\def\BibTeX{{\rm B\kern-.05em{\sc i\kern-.025em b}\kern-.08em
    T\kern-.1667em\lower.7ex\hbox{E}\kern-.125emX}}

\begin{document}

\newtheorem{theorem}{Theorem}
\newtheorem{lemma}[theorem]{Lemma}
\newtheorem{corollary}[theorem]{Corollary}
\newtheorem{proposition}[theorem]{Proposition}
\newtheorem{exercise}[theorem]{Exercise}
\newtheorem{definition}[theorem]{Definition}
\newtheorem{conjecture}[theorem]{Conjecture}
\newtheorem{observation}[theorem]{Observation}

\theoremstyle{remark}
\newtheorem{example}[theorem]{Example}
\newtheorem{note}[theorem]{Note}
\newtheorem{remark}[theorem]{Remark}

\bibliographystyle{plainurl}
\title{Distal combinatorial tools for graphs of bounded twin-width}

\author{\IEEEauthorblockN{Wojciech Przybyszewski \thanks{
    \vspace*{0.2cm}
    \hspace*{-2cm}\begin{minipage}{\columnwidth}
        This work is a part of project {\sc{BOBR}} that has received funding from the European Research Council (ERC) under the European Union's Horizon 2020 research and innovation programme (grant agreement No 948057).
    \end{minipage}\hfill
    \begin{minipage}{0.5\columnwidth}
        \includegraphics[width=\textwidth]{images/ERC.jpg}
    \end{minipage}\hfill
}}
\IEEEauthorblockA{Institue of Informatics\\
Univeristy of Warsaw\\
Warsaw, Poland}}

\IEEEoverridecommandlockouts
\IEEEpubid{\makebox[\columnwidth]{979-8-3503-3587-3/23/\$31.00~
\copyright2023 IEEE \hfill} \hspace{\columnsep}\makebox[\columnwidth]{ }}

\maketitle

\begin{abstract}
We study set systems formed by neighborhoods in graphs of bounded twin-width.
We start by proving that such graphs have linear neighborhood complexity, in analogy to previous results concerning graphs from classes with bounded expansion and of bounded clique-width.
Next, we shift our attention to the notions of distality and abstract cell decomposition, which come from model theory.
We give a direct combinatorial proof that the edge relation is distal in classes of ordered graphs of bounded twin-width.
This allows us to apply Distal cutting lemma and Distal regularity lemma, so we obtain powerful combinatorial tools for graphs of bounded twin-width.
\end{abstract}

\begin{IEEEkeywords}
    graph theory, twin-width, distality, neighborhood complexity,  Vapnik-Chervonenkis density, abstract cell decomposition, regularity lemma, cutting lemma

\end{IEEEkeywords}

\section{Conclusions}
\label{chap:conclusions}

One of the most important conjectures in structural and algorithmic graph theory is that model checking is FPT on a class of graphs if and only if the class is monadically NIP \cite{conjecture-nip}.
As we mentioned, this notion comes from model theory and therefore it is believed that progressing towards this conjecture requires a better understanding of the interplay of model theory and structural graph theory.
We believe that studying the notion of distality is a step in this direction.

We can observe the mentioned interplay when we can prove a theorem in two completely different ways -- one using tools from model theory, and the other being a direct, combinatorial proof.
An example of such theorem is shown in \cite{flipper}.
We are grateful to the Anonymous Reviewer who pointed out that Theorem \ref{theorem:regularity-lemma-tww} alternatively can be proven in a purely combinatorial way using \cite[Lemma 20]{twinwidth2}.
We believe that studying connections between these two proofs is an interesting goal.

As we mentioned in Section \ref{chap:encoding} a class of ordered graphs is distal if and only if it has bounded twin-width.
However, the distal regularity lemma and the distal cutting lemma does not only apply to graphs of bounded twin-width.
As an example, we can consider classes of graphs of bounded expansion equipped with an arbitrary order.
They are not distal, but every first-order formula which doesn't use order is distal in them \cite{distal-be}.
This leads to a possible generalization of classes of graphs of bounded expansion and bounded twin-width and such a generalization is highly anticipated and sought.
In light of these facts, Theorem \ref{theorem:cutting-lemma-tww} and Theorem \ref{theorem:regularity-lemma-tww} may be seen as a step towards finding common combinatorial properties of the two notions.

\section*{Acknowledgments}
I am deeply grateful to Prof. Szymon Toruńczyk, my master's thesis supervisor, for introducing me to this area of research as well as for his assistance and dedicated involvement in every step throughout the process of preparing this paper.
I am also grateful to the Anonymous Reviewers for their insightful comments, which helped me in rewriting this paper to its final form.

\IEEEtriggeratref{31}
\bibliography{bibliography}


\end{document}
\fi